\documentclass[11pt]{article}

\usepackage{booktabs}

\usepackage{microtype}
\usepackage{times}
\usepackage{titlesec}

\usepackage{latexsym,amsmath,amssymb}
\usepackage{amsthm}
\usepackage{color,soul}
\usepackage{ifthen}
\usepackage{xspace}
\usepackage{paralist}
\usepackage{xstring}

\usepackage{fullpage}

\usepackage{ifpdf}

\ifpdf
\usepackage[pdftex]{graphicx}
\usepackage[pdftex,
  ocgcolorlinks,
  linkcolor=blue,
  filecolor=blue,
  citecolor=blue,
  urlcolor=blue]{hyperref}
\else
\usepackage[dvips]{graphicx}
\fi

\DeclareMathOperator*{\Exp}{\mathbb{E}}
\DeclareMathOperator{\Ima}{Im}

\usepackage[ruled]{algorithm2e} 

\SetAlFnt{\small}
\SetAlCapFnt{\small}
\SetAlCapNameFnt{\small}
\SetAlCapHSkip{0pt}
\IncMargin{-\parindent}

\begin{document}

\title{Multi-Party Protocols, Information Complexity and Privacy}
\author{
Iordanis Kerenidis\thanks{CNRS and Universit\'e Paris Diderot, email: {\tt jkeren@irif.fr}.}
\and 
Adi Ros\'en\thanks{CNRS and Universit\'e Paris Diderot, email: {\tt adiro@irif.fr}.}
\and
Florent Urrutia\thanks{ Universit\'e Paris Diderot, email: {\tt urrutia@irif.fr}.}
}

\newcommand{\PIC}{\mathsf{PIC}}

\date{}

\maketitle

\begin{abstract}

We introduce a new information theoretic measure that we call {\em Public Information Complexity} ($\PIC$), as a tool for the study of multi-party computation protocols, and of quantities such as their communication complexity, or the amount of randomness they require in the context of information-theoretic private computations.
We are able to use this measure directly in the natural asynchronous message-passing {\em peer-to-peer} model and show a  number of interesting properties and applications of our new notion:
the Public Information Complexity is a lower bound on the Communication Complexity and an upper bound on the Information Complexity;
the difference between the Public Information Complexity and the Information Complexity provides a lower bound on the amount of randomness used in a protocol; 
any communication protocol can be compressed to its Public Information Cost; an explicit  calculation of  the zero-error Public Information Complexity of the $k$-party, $n$-bit Parity function, where a player outputs the bit-wise parity of the inputs. The latter result also establishes that the amount of randomness needed by a private protocol that computes this function is $\Omega(n)$. 

\end{abstract}

\theoremstyle{plain}
\newtheorem{theorem}{Theorem}[section]
\newtheorem{proposition}[theorem]{Proposition}
\newtheorem{fact}[theorem]{Fact}
\newtheorem{definition}[theorem]{Definition}
\newtheorem{lemma}[theorem]{Lemma}
\newtheorem{corollary}[theorem]{Corollary}

\newcommand{\zo}{\{0,1\}}
\newcommand{\eps}{\epsilon}
\newcommand{\ie}{\emph{i.e.} }
\newcommand{\GF}{\mathbb{F}}
\newcommand{\IC}{\mathsf{IC}}
\newcommand{\EC}{\mathsf{IC}^{\text{ext}}}
\newcommand{\CC}{\mathsf{CC}}
\newcommand{\ACC}{\mathsf{ACC}}
\newcommand{\R}{\mathsf{R}}
\newcommand{\PEC}{\mathsf{PIC}^{\text{ext}}}
\newcommand{\AND}{\mathsf{AND}}
\newcommand{\XOR}{\mathsf{Par}}
\newcommand{\Ber}{\text{\textbf{Ber}}}
\newcommand{\calA}{\mathcal{A}}
\newcommand{\calB}{\mathcal{B}}
\newcommand{\calD}{\mathcal{D}}
\newcommand{\calX}{\mathcal{X}}
\newcommand{\calY}{\mathcal{Y}}
\newcommand{\calO}{\mathcal{O}}
\newcommand{\calL}{\mathcal{L}}
\newcommand{\calM}{\mathcal{M}}

\newcommand{\Ni}[2]{\ensuremath{[\![#1,#2]\!]}}

\newcommand{\adi}[1]{Adi: #1}

\newcommand{\M}[3]{
    \IfEqCase{#3}{
        {r}{M_{#1}^{\text{\makebox[0mm]{\,\,$\overleftarrow{\color{white}l}$}}#2}}
        {s}{M_{#1}^{\text{\makebox[0mm]{\,\,$\overrightarrow{\color{white}l}$}}#2}}
    }[\PackageError{M}{Undefined option to M: #2}{}]
}
\newcommand{\T}[3]{
    \IfEqCase{#3}{
        {r}{M_{#1}^{<\text{\makebox[0mm]{\,\,$\overleftarrow{\color{white}l}$}}#2}}
        {s}{M_{#1}^{<\text{\makebox[0mm]{\,\,$\overrightarrow{\color{white}l}$}}#2}}
    }[\PackageError{T}{Undefined option to T: #2}{}]
}

\section{Introduction}
Communication complexity, originally introduced by Yao~\cite{Yao}, is a prolific field of research in theoretical computer science that yielded many important results in various fields. 
Informally, it attempts to answer the question ``How many bits must distributed  players transmit to solve a given distributed problem ?''  The study of the two-party case has produced a large number of interesting and important results, both upper and lower bounds, with many applications in other areas in theoretical computer science such as circuit complexity, data structures, streaming algorithms and distributed computation (see, e.g.,~\cite{KN, MNSW, GG, SHKKNPPW, FHW}). 

A  powerful tool recently introduced for the study of two-party communication protocols is the measure of {\em Information Complexity} (or \textit{cost}). 
This measure, originally defined in~\cite{Bar-YehudaCKO93} and~\cite{CSWY}, extends the notions of information theory, originally introduced by Shannon~\cite{Sha}, to  interactive settings. It quantifies, roughly speaking, the amount of information about their respective inputs that Alice and Bob must leak to each other in order to compute a given function $f$ of their inputs.
Information complexity (IC) has been used in a long series of papers  to prove lower bounds on communication complexity and other properties of (two-party) communication protocols (e.g., \cite{BYJKS,BBCR,BR,Bra}).
An interesting property of information complexity is that it satisfies a direct sum property. 
The {\em direct sum} question, one of the most interesting questions in complexity theory, asks whether solving $n$ independent copies of the same problem must cost (in a given measure) $n$ times the cost of solving a single instance. In the case of communication complexity, this question has been studied in, e.g.,  \cite{FKNN, CSWY, Shaltiel, JRS, HJMR, BBCR, Kla, Jain} and in many cases it remains open whether a direct sum property holds.

Another important question in communication complexity is the relation between the information complexity of a function and its communication complexity. One would like to know if it is possible to compute a function by sending a number of bits which is not (too much) more than the information the protocol actually has to reveal. Put differently, is it always possible to {\em compress} the communication cost of a protocol to its information cost? For the two-party case it is known that  perfect compression is not possible for single shot protocols~\cite{GKR3,GKR2} (unless they are restricted to a constant number of rounds~\cite{JRS}). Still, several interesting compression results are known. The equality between information cost and \emph{amortized} communication cost is shown in \cite{BR,Bra}, and other compression techniques  are given in \cite{BBCR,BMY,BBKLS,Pan}. It remains open if one can compress interactive communication up to some small loss (for example logarithmic in the size of the input). The specific case of compression under product distributions was studied in \cite{Kol,She}, leading to a compression to $\calO(I~\text{polylog($I$)})$.

When trying to study the {\em multi-party} (i.e., where at least $3$ players are involved) communication settings using similar information-theoretic methods, such as IC, one encounters a serious problem. The celebrated results on information-theoretic private computation~\cite{BGW,CCD} state that if the number of players is at least $3$, then {\em any  function} can be computed by a randomized protocol such that {\em no information} about the inputs is revealed to the other players (other than what is implied by the value of the function and their own input). Thus, in the multi-party case, the $\IC$ of any function $f$ is $0$ (or only the entropy of $f$, depending on the definition of $\IC$), and cannot serve to study multi-party protocols.

For this reason, information theory has rarely been used in the multi-party setting, where most results related to communication complexity have been obtained via combinatorial techniques. Among the interesting works on multi-party settings are \cite{PVZ,WZ} which introduce the techniques of \textit{symmetrization} and \textit{composition}, and \cite{CRR,CR} which study the influence of the topology of the network. One notable exception is the work of Braverman et al.~\cite{BEOPV} which studies the \textit{set-disjointness} problem using information theoretic tools. Braverman et al. provide almost tight bounds in the so-called coordinator model (that differs from the more natural peer-to-peer model) by analyzing the information leaked between the players but also the information obtained by the coordinator itself. The set disjointness problem is arguably one of the most extensively studied problem in communication complexity (cf.~\cite{Bra,BYJKS,CKS,Gro,Jay,BGPW}). This line of research  was followed by \cite{CM} which also uses information theory to obtain tight bounds on the communication complexity of the function \textit{Tribes} in the coordinator model. Information theory is also used in \cite{BO} to study set-disjointness in the broadcast model. A compression procedure for the broadcast model is described in \cite{KOS}.

A number of sub-models have been considered in the literature for the multi-party computation protocols setting: the \textit{number in  hand model} (NIH), where each player has a private input, is arguably the most natural one, while in the 
\textit{number on the forehead model} (NOF), each player $i$ knows all inputs $x_j$, $j \neq i$, i.e., the ``inputs'' of all players except its own. As to the communication pattern, a number of variants have been considered as 
well: in the \textit{blackboard} model, the players communicate by broadcasting messages (or writing them on a ``blackboard''); in the \textit{message passing} model, each pair of players is given a private channel to mutually communicate (for more details on the different variants see \cite{KN}). Most of the results obtained in multi-party communication complexity were obtained for the NOF model or the blackboard model. The present paper studies, however, the NIH, message passing
(peer to peer)  model, which is also the most closely related to the work done on message passing protocols in the distributed computing and networking communities.

\subsection{Our contributions}

Our main goal is to introduce novel information-theoretical measures for the study of number-in-hand, message-passing multi-party protocols, coupled with a natural model that, among other things,  allows private protocols (which is not the case for, e.g.,  the coordinator model).

We define the new measure of {\em Public Information Complexity} ($\PIC$), as a tool for the study of multi-party computation protocols, and of quantities such as their communication complexity, or the amount of randomness they require in the context of information-theoretic private computations.
Intuitively, our new measure captures a combination of the amount of information about the inputs that the players leak to other players, and the amount of randomness that the protocol uses. By proving lower bounds on $\PIC$ for a given multi-party function $f$, we are able to give lower bounds on the 
multi-party communication complexity of $f$ and on the amount of randomness needed to privately compute $f$. The crucial point is that the $\PIC$ of functions, in our multi-party model, is not always $0$, unlike their $\IC$. 

Our new measure works in a model which is a slight restriction of the most general asynchronous model, where, for a given player at a given time, the set of players from which that player waits for a message can be determined by that player's own local view.
 This allows us to have the property that for any protocol, the information which is  leaked  during the execution of the protocol is at most the communication cost  of the protocol.  Note that in the multi-party case, the information cost of a protocol may be higher than its communication cost, because the identity of the player from which one receives a message might carry some information. 
We are able to define our measure and use it directly in a natural asynchronous {\em peer-to-peer} model (and not, e.g., in the coordinator model used in most works studying the multi-party case, c.f.~\cite{DF}). 
The latter point is particularly important when one is interested in private computation, since our model allows for private protocols, while this is not necessarily the case for other models.
Furthermore, if one seeks to accurately understand  communication complexity in the natural peer-to-peer model, suppressing 
polylog-factor inaccuracies, one has to study directly the peer-to-peer model, because lower bounds in the coordintor model translate to the peer-to-peer model only up to logarithmic factors (see the comparison of models in subsection~\ref{subsec model}).

We then continue and  show a  number of interesting properties and applications of our new notion:
\begin{itemize}
\item The Public Information Complexity is a lower bound on the Communication Complexity and an upper bound on the Information Complexity. In fact, it can be strictly larger than the Information Complexity.
\item The difference between the Public Information Complexity and the Information Complexity provides a lower bound on the amount of randomness used in a protocol. 
\item We compress any communication protocol to their $\PIC$ (up to logarithmic factors), by extending  to the multi-party setting the work of Brody et al.~\cite{BBKLS} and Pankratov~\cite{Pan}.
\item We show that one can approach the central question of direct sum in communication complexity by trying to prove a direct sum result for $\PIC$. Indeed, we show that a direct sum property for $\PIC$ implies a certain direct sum property for communication
complexity.
\item We precisely calculate the zero-error \textit{Public Information Complexity} of the $k$-party, $n$-bit Parity function ($\XOR$), where a player outputs the bit-wise parity of the inputs. We show that the $\PIC$ of this function is $n(k-1)$. This result is tight and it also establishes that the amount of randomness needed for a private protocol that computes this function is $\Omega(n)$. While this sounds a reasonable assertion no previous proof for such claim existed. 
\end{itemize}

\subsection{Organization}

The paper is organized as follows. In section~\ref{sec:preliminaries} we review several notations and information theory basics. In Section~\ref{sec:model} we define the communication model that we work with and a number of traditional complexity measures. In Section~\ref{sec:pic} we define the new measure $\PIC$ that we introduce in the present paper, and in Section~\ref{sec:privacy} we discuss its relation to randomness and multi-party private computation. In Section~\ref{sec:xor} we give tight bounds for the parity function $\XOR$, using $\PIC$. In section~\ref{sec:direct-sum}, we discuss the existence of a direct sum property for $\PIC$.

\section{Preliminaries}
\label{sec:preliminaries}
We start by defining a number of notations. 
We denote by $k$ the number of players. We often use $n$ to denote the size (in bits) of the input to each player.
Calligraphic letters will be used to denote sets. Upper case letters will be used to denote random variables, and given two random variables $A$ and $B$, we will denote by $AB$ the joint random variable $(A,B)$. Given a string (of bits) $s$, $|s|$ denotes the length of $s$. Using parentheses we denote an ordered set  (family) of items, e.g., $(Y_i)$. Given a family $(Y_i)$,  $Y_{-i}$ denotes the sub-family which is the family $(Y_i)$ {\em without} the element $Y_i$. The letter $X$ will usually denote the input to the players, and we thus use the shortened notation $X$ for $(X_i)$, \textit{i.e.} the input to all players. $\pi$ will be used to denote a protocol. We use the term {\em entropy} to talk about binary entropy.

We give a reminder on basic information theory, as introduced in \cite{Sha}.

\begin{definition}
The \textit{entropy} of a (discrete) random variable $X$ is
$$H(X) = \sum\limits_{x}\Pr[X=x]\log\left(\frac{1}{\Pr[X=x]}\right).$$
The \textit{conditional entropy} $H(X \mid Y)$ is defined as $\Exp\limits_y[H(X \mid Y=y)]$.
\end{definition}

\begin{proposition}\label{H b}
For any finite set $\calX \subseteq \{0,1\}^{*}$ and any random variable $X$ with support $\text{supp}(X) \subseteq \calX$, it holds
$$H(X) \le \log( | \calX |).$$
Moreover, if the set $\calX$ is prefix-free, it holds $H(X) \le \Exp[ | X |]$.

\end{proposition}
\begin{definition}
The \textit{mutual information between two random variables $X,Y$} is
$$I(X;Y) = H(X) - H(X \mid Y).$$
The \textit{conditional mutual information} $I(X;Y \mid Z)$ is $H(X \mid Z) - H(X \mid YZ)$. 
\end{definition}

The mutual information measures the change in the  entropy of $X$ when one learns the value of $Y$. It is non negative, and is symmetric: $I(X;Y) = I(Y;X)$.

\begin{proposition}\label{indp}
For any random variables $X$, $Y$ and $Z$, $I(X;Y \mid Z) = 0$ if and only if $X$ and $Y$ are independent for each possible value of $Z$.
\end{proposition}
\begin{proposition}\label{condi}
For any random variables $X$ and $Y$, $H(X \mid Y) \le H(X)$.
\end{proposition}
\begin{proposition}[Chain Rule]\label{CR}
Let $A$, $B$, $C$, $D$ be four random variables. Then
$$I(AB ; C \mid D) = I(A ; C \mid D) + I(B ; C \mid DA).$$
\end{proposition}
\begin{lemma}[Data processing inequality]\label{data proc}
For any $X$, $Y$, $Z$, and any function $f$, it holds: $I(X ; f(Y) \mid Z)\le I(X ; Y \mid Z)$
\end{lemma}
\begin{proof}
\begin{align*}
I(X ; f(Y) \mid Z) &\le I(X ; f(Y) \mid Z) + I(X ; Y \mid f(Y) Z) \\
&= I(X ; Y f(Y) \mid Z) \\
&= I(X ; Y \mid Z) + I(X ; f(Y) \mid Y Z) \\
& =  I(X ; Y \mid Z).
\end{align*}
\end{proof}
\begin{proposition}[\cite{Bra}]\label{2.9}
Let $A$, $B$, $C$, $D$ be four random variables such that $I(B ; D \mid AC) = 0$. Then
$$I(A;B \mid C) \ge I(A ; B \mid CD).$$
\end{proposition}
\begin{proposition}[\cite{Bra}]\label{2.10}
Let $A$, $B$, $C$, $D$ be four random variables such that $I(B ; D \mid C) = 0$. Then
$$I(A;B \mid C) \le I(A ; B \mid CD).$$
\end{proposition}

\section{The model}
\label{sec:model}
We now define a natural communication model which is a slight restriction of the most general asynchronous peer-to-peer model. 
Its restriction is that for a given player at a given time, the set of players from which that player waits for a message can be determined by that player's own local view. The player continues its computation only after messages are received
 from this set.
This allows us to define information theoretical tools that pertain to the transcripts of the protocols, and at the same time to use these tools as lower bounds for communication complexity. This restriction however does not exclude the existence of private protocols, as other special cases of the general asynchronous  model do.  We observe  that without such restriction the information revealed by the execution of a protocol might be higher than the number of bits transmitted and that, on the other hand, practically all multi-party protocols in the literature 
are implicitly defined in our model. We also compare our model to the general one and to other restricted ones and explain the usefulness and logic of our specific model. 

\subsection{Definition of the model}\label{subsec model}
We work in the \textit{multi-party number in hand peer-to-peer} model. Each player has unbounded local computation power and, in addition to its input $X_i$, has access to a source of private randomness $R_i$. We will use the notation $R$ for $(R_i)$, \textit{i.e.}, the private randomness of all players. A source of public randomness $R^p$ is also available to all players.
The system consists of $k$ players and a family of $k$ functions $f = (f_i)_{i\in\Ni{1}{k}}$, with $\forall~i \in\Ni{1}{k}, f_i: \prod\limits_{l=1}^{k}\calX_l \rightarrow \calY_i$, where $\calX_l$ denotes the set of possible inputs of player $l$, and  ${\calY_i}$  denotes the set of possible outputs of player $i$.
The players are given some input $x = (x_i)\in \prod\limits_{i=1}^{k}\calX_i$, and for every $i$, player $i$ has to compute $f_i(x)$. 
Each player has a special write-only output tape.

We define the communication model as follows, which is the asynchronous setting, with some restrictions.
To make the discussion simpler we assume a global time which is {\em unknown} to the players.
Every pair of players is connected by a bidirectional communication link that allows them to send
messages in both directions. There is no bound on the delivery time (i.e., when the message arrives to its destination node) of a message, but every message
is delivered in finite time, and the communication link maintains FIFO order in each of the two directions.
Messages that arrive to the head of the link at the destination node of that link are buffered until they are read by that node.
Given a specific time we define the {\em view} of player $i$, denoted $D_i$, as the input of that player, $X_i$, its private randomness, $R_i$, and the messages received so far by player $i$. The protocol of each player $i$ runs in {\em local} rounds.
In each round, player $i$ sends messages to some subset of the other players. The identity of these players, as well as the content of these messages, depend on the current view of player $i$.  The player also decides whether to write a (nonempty) string on its output tape. Then, the player waits for messages from a certain subset of the other players, where this subset is also determined by the current view of the player. That is,  the player reads a single message from each of the incoming links that connect it  to that subset of other players. If for a certain such link no message
is available, then the player waits until such message is available (i.e., arrives).
Then the (local) round of player $i$ terminates\footnote{The fact that the receiving of the incoming messages comes as the last step of  the (local) round comes only to emphasize that the sending of the messages and the writing on the output tape are a function of only the messages received in previous (local) rounds.}. To make it possible for the player to identify the arrival of the {\em complete} message that it waits for, we require that each message sent by a player in the protocol be self-delimiting.
  
 Denote by ${\calD}_i^j$ the set of possible views of player $i$ at the end of local round 
$j$, $j \geq 0$, where the beginning of 
the protocol is considered round $0$. 
Formally, a  protocol $\pi$ is defined by  a set of local programs, one for each player $i$,  where the local 
program of player $i$ is defined by a sequence of functions, parametrized by the index of the {\em local} 
round $j$, $j\geq 1$:

\begin{itemize}
\item $\overline{S}_i^{j}: {\calD}_i^{j-1} \rightarrow 2^{\{1,\ldots,k\}\setminus\{i\}}$, defining the set of players 
to which player $i$ \textit{sends} the messages.
\item $m_{i,q}^j:  {\calD}_i^{j-1} \rightarrow \{0,1\}^{*}$, such that for any $D_i^{j-1} \in{\calD}_i^{j-1}$, 
if  $q \in \overline{S}_i^{j}(D_i^{j-1})$,  then $m_{i,q}^j(D_i^{j-1})$ is 
the content of the message player $i$ sends to player $q$. Each such message is self-delimiting.
\item $O_i^j: \calD_i^{j-1} \rightarrow \{0,1\}^{*}$, defining what the player writes  on the output tape. Each player can write on 
its output tape a non-empty string only once.\footnote{We require that each player writes only once on its output tape so 
that the local view of the player determines the local output of the protocol (i.e., so that the player itself ``knows'' the output). 
This requirement is needed since  a player may not know locally that the protocol ended.}
\item $S_i^{j}: {\calD}_i^{j-1} \rightarrow 2^{\{1,\ldots,k\}\setminus\{i\}}$, 
defining the set of players from which player $i$ waits to \textit{receive} a message.
\end{itemize}

We define the transcript of the protocol of player $i$, denoted $\Pi_i$, as the concatenation of the messages read by player $i$ from the links of the sets $S_i^1, S_i^2, \ldots$, ordered by local round number, and within each round by the index of the player. We denote by 
$\overleftrightarrow{_{~}\Pi_i}$ the concatenation of $\Pi_i$ together with a similar concatenation $\overrightarrow{_{~}\Pi_i}$ of the messages sent by player $i$ to the sets $\overline{S}_i^0, \overline{S}_i^1,\ldots$
We denote by  $\Pi_{i \rightarrow j}$ the concatenation of the messages sent by player $i$ to player $j$ during the course of the protocol. The transcript of the (whole) protocol, denoted $\Pi$, is obtained by concatenating all the $\Pi_i$ ordered by, say, player index. 

We will give most of the definitions for the case where all functions $f_i$ are the same function, that we denote by $f$. The definitions in the case of family of functions are similar.

\begin{definition}
For $\epsilon \ge 0$, a protocol $\pi$ $\epsilon$\textit{-computes} a function $f$ if 
for all $(x_1,\dots,x_k)\in\calX_1\times\dots\times\calX_k$:

\begin{enumerate}
\item  For all possible assignments for the random sources $R_i$, $ 1\leq i\leq  k$, and $R^p$,  every player eventually (i.e., in finite time) writes on its output tape (a non-empty string).
\item With probability at least $1-\eps$  (over all random sources) the following event occurs:
each player $i$ writes on its output tape the value $f(x)$, i.e., the correct value of the function.
\end{enumerate}
\end{definition}
 
For simplicity we also assume that a protocol must eventually stop. That is, for all possible inputs and all possible assignments for the random sources, eventually (i.e., in finite time) there is no message in transit.

\subsection{Comparison to other models}
The somewhat restricted model (compared to the general asynchronous model) that we work with allows us to define a measure similar to information cost that we will later show to have desirable properties and to be of use.
Notice that the general asynchronous model is problematic in this respect since one bit of communication can bring $\log(k)$ bits of information, as not only the content of the message but also the identity of the sender may reveal information. Thus, information cannot be used as a lower bound on communication. In our case, the sets $S_i^l$ and $\overline{S}_i^l$ are determined by the current view of the player, $(\Pi_i)$ contains only the content of the messages, and thus the desirable relation between the communication and the information is maintained. On the other hand, our restriction is natural, does not seem to be very restrictive (practically all protocols in the literature adhere to our
model), and does not exclude the existence of private protocols.

To exemplify the above mentioned issue in the general asynchronous  model consider the following simple example of a deterministic protocol, for $4$ players $A$, $B$ and $C$, $D$, which allows $A$ to transmit to $B$ its input bit $x$, but where all messages sent in the protocol are the bit $0$, and the protocol generates only a single transcript over all possible inputs.

{\bf A}: If $x=0$ send $0$ to $C$; after receiving 0 from $C$, send $0$ to $D$.

\hspace{0.35cm} If $x=1$ send $0$ to $D$; after receiving 0 from $D$, send $0$ to $C$

{\bf B}: After receiving 0 from a player, send 0 back to that player.

{\bf C,D}: After receiving 0 from $A$ send 0 to $B$. After receiving 0 from $B$ send 0 to $A$.

\noindent It is easy to see that $B$ learns the value of $x$ from the order of the messages it gets.

There has been a long series of works about multi-party communication protocols in different variants of models, for example  \cite{CKS,Gro,Jay,PVZ,CRR,CR}. 
In \cite{BEOPV}, Braverman et al. consider a restricted class of protocols working in the \textit{coordinator model}: an additional player with no input can communicate privately with each player, and the players can only communicate with the coordinator. 
  
We first note that the coordinator model does not yield exact bounds for the multi-party communication complexity in the peer-to-peer model (neither in our model nor in the most general one). 
Namely, a protocol in the peer-to-peer model can be transformed into a protocol in the coordinator model with 
an $O(\log k)$ multiplicative factor in the communication complexity, by sending any message to the coordinator with a $O(\log k)$-bit label indicating its destination. This factor is sometimes necessary, e.g., for the {\tt $q$-index} function, where players  $P_i$, $0 \leq i \leq  k-1$, each holds an input bit  $x_i$,  player $P_k$ holds $q$ indices $ 0 \leq j_{\ell} \leq k-1$,  $1\leq \ell \leq q$, and $P_k$ should learn the vector $(x_{j_1},x_{j_1},\ldots,x_{j_{q}})$: in the coordinator model the communication complexity of this function is $\Theta(\min\{k,q\log k\})$, while in both peer-to-peer models there is a protocol for this function that sends only (at most) $\min\{k,2q\}$ bits, where $P_k$ just queries the appropriate other players. But this multiplicative factor between the complexities in the two models is not always necessary: the communication complexity of the parity function $\XOR$ is $\Theta(k)$ both in the peer-to-peer models and in the coordinator model. 

Moreover, when studying private protocols in the peer-to-peer model, the coordinator model does not offer any insight. In the (asynchronous) coordinator model, described in \cite{DF} and used for instance in \cite{BEOPV}, if there is no privacy requirement with respect to the coordinator, it is trivial to have a private protocol by all players sending their input to the coordinator, and the coordinator returning the results to the players. If there is a privacy requirement with respect to the coordinator, then if there is a random source shared by all the players (but not the coordinator), privacy is always possible using the protocol of~\cite{FKN}. If no such source exists, privacy is impossible in general.
This follows from the results of Braverman et al. \cite{BEOPV} who show a non-zero lower bound on the total internal information complexity of all parties (including the coordinator) for the function {\em Disjointness} in that model. 

Note also that the private protocols described in \cite{BGW,CCD} (and further work) are defined in the synchronous setting, and thus can be adapted to our communication model (the sets $\overline{S}_i^j$ and ${S_i^j}$ are always all the players and hence even independent of the current views). 

In the sequel we also use a special case of our model, where the sets $\overline{S}_i^j$ and  $S_i^j$ are a function only of $i$ and $j$, and not of the entire current view of the player. This is a natural special case for protocols which we call {\em oblivious protocols}, where the communication pattern is fixed and is not a function of the input or randomness. Clearly, the messages themselves remain a function of the view of the players. 
We observe that synchronous protocols are a special case of oblivious protocols.

\subsection{Communication complexity and information complexity}

Communication complexity, introduced in \cite{Yao}, measures how many bits of communication are needed in order for a set of  players to compute with error $\epsilon$ a given function of their inputs. The allowed error $\epsilon$, implicit in many of the contexts, will be written explicitly as a superscript when necessary.

\begin{definition}
The \textit{communication cost} of a protocol $\pi$, $\CC(\pi)$, is the maximal length of the transcript of $\pi$ over  all possible inputs, private randomness and public randomness.
\end{definition}

\begin{definition}
$\CC(f)$ denotes the communication cost of the best protocol computing $f$.
\end{definition}

Information complexity measures the amount of information that must be transmitted so that the players can compute a given function of their joint inputs.  One of its main uses is to provide  a lower bound on the communication complexity of the function. In the two-party setting, this measure
led to interesting results on the communication complexity of various functions, such as {\em AND} and {\em Disjointness}. We now focus on designing an  analogue to the information cost, for the multi-party setting. The notion of internal information cost for two-party protocols (c.f. \cite{CSWY,BYJKS,Bra}) can be easily generalized to any number of players:

\begin{definition}
The \textit{internal information cost} of a protocol $\pi$ for $k$ players, with respect to input distribution $\mu$, is the sum of the information revealed to each player about the inputs of the other players: 
$$\IC_\mu(\pi) = \sum\limits_{i=1}^k I(X_{-i} ; \Pi_i \mid X_i R_i R^p).$$
\end{definition}

Intuitively, the information cost of a protocol is the amount of information each player learns about the inputs of the other players during the protocol. The definition we give above, when restricted to two players is the same as in \cite{Bra}, even though they look slightly different. This is because we make explicit the role of the randomness, which will allow us to later give bounds on the amount of randomness needed for private protocols in the multi-party setting.

The \textit{internal information complexity} of a function $f$ with respect to input distribution $\mu$, as well as the \textit{internal information complexity} of a function $f$, can be defined for the multi-party case based on the information cost of a protocol, just as in the $2$-party case. 

\begin{definition}
The \textit{internal information complexity} of a function $f$, with respect to input distribution $\mu$ is the infimum of the internal information cost over all protocols computing $f$ on input distribution $\mu$:
$$\IC_\mu(f) =  \inf\limits_{\pi\text{ computing }f} \IC_\mu(\pi).$$
\end{definition}

The information revealed to a given player by a protocol can be written in several ways:

\begin{proposition}\label{eq def}
For any protocol $\pi$, for any player $i$:
$$I(X_{-i} ; \overleftrightarrow{_{~}\Pi_i} \mid X_i R_i R^p) = I(X_{-i} ; \Pi_i \mid X_i R_i R^p).$$
\end{proposition}

\begin{proof}
For any protocol $\pi$, for any player $i$:
\begin{align*}
I(X_{-i} ; \overleftrightarrow{_{~}\Pi_i} \mid X_i R_i R^p) &= I(X_{-i} ; \overrightarrow{_{~}\Pi_i}\Pi_i \mid X_i R_i R^p) \\
&=I(X_{-i} ;\Pi_i \mid X_i R_i R^p)  + I(X_{-i} ;\overrightarrow{_{~}\Pi_i} \mid X_i R_i R^p \Pi_i) \text{~~~(chain rule)}\\
&= I(X_{-i} ;\Pi_i \mid X_i R_i R^p) \text{~~~(since $H(\overrightarrow{_{~}\Pi_i} \mid X_i R_i R^p \Pi_i)=0$)}~.
\end{align*}

\end{proof}

\subsection{Information complexity and privacy}

The definition of a {\em private protocol} as defined in~\cite{BGW,CCD} is the following.
\begin{definition} \label{def:privacy}
A $k$-player protocol $\pi$ for computing a family of functions $(f_i)$ is \textit{private}\footnote{In this paper we consider only the setting of $1$-privacy, which we call here for simplicity, privacy.} if for every player $i \in \Ni{1}{k}$,
for all pairs of inputs $x=(x_1,\dots,x_k)$ and $x'=(x'_1,\dots,x'_k)$ such that $f_i(x) = f_i(x')$ and $x_i = x'_i$,
for all possible private random tapes $r_i$ of player $i$, and all possible  public random tapes $r^p$, it holds that for any transcript $T$ 
$$\Pr[\Pi_i = T \mid R_i=r_i~;~X=x~;~R^p=r^p] = Pr[\Pi_i = T \mid R_i=r_i~;~X=x'~;~R^p=r^p]~,$$ where the probability is over the randomness $R_{-i}$.
\end{definition}

The notion of privacy has an equivalent formulation in terms of information.
\begin{proposition}\label{p05}
A protocol $\pi$ is private if and only if for all input distributions $\mu$,
$$\sum\limits_{i=1}^k I(X_{-i} ; \Pi_i \mid X_i R_i R^p f_i(X)) = 0.$$
\end{proposition}
\begin{proof}
By proposition \ref{indp}, definition~\ref{def:privacy} is equivalent to the following: 
$$\forall~i, I(X_{-i} ; \Pi_i \mid X_i R_i R^p f_i(X)) = 0~.$$
Since $I$ is non-negative, this is equivalent to 
$$\sum\limits_{i=1}^k I(X_{-i} ; \Pi_i \mid X_i R_i R^p f_i(X)) = 0~.$$
\end{proof}
It is well known that in the multi-party number-in-hand peer-to-peer setting (for $k \geq 3$), unlike in the two-party case, {\em any} function can be privately computed. 

\begin{theorem}[\cite{BGW},\cite{CCD}]
Any family of functions of more than two variables can be computed by a private protocol.
\end{theorem}
Using the above theorem, we can give the following lemma.
\begin{lemma}\label{p1}
For any family of functions $(f_i)$ of more than two variables and any 
$\mu$, \\$\IC_\mu(f) \le \sum\limits_{i=1}^k H(f_i(X)),$ where $X$ is distributed according to $\mu$.
\end{lemma}

\begin{proof}
Let $\pi$ be a $k$-player private protocol computing $(f_i)$. Fix a distribution $\mu$ on the inputs.
\begin{align*}
\IC_{\mu}(\pi) &= \sum\limits_{i=1}^k I(X_{-i} ; \Pi_i \mid X_i R_i R^p)\\
&\le \sum\limits_{i=1}^k I(X_{-i} ; \Pi_i f_i(X) \mid X_i R_i R^p)\\
&= \sum\limits_{i=1}^k \left[I(X_{-i} ; f_i(X) \mid X_i R_i R^p) + I(X_{-i} ; \Pi_i \mid X_i R_i R^p f_i(X))\right]\\
&= \sum\limits_{i=1}^k I(X_{-i} ; f_i(X) \mid X_i R_i R^p)\\
&\le \sum\limits_{i=1}^k H(f_i(X)).
\end{align*}
Now, $\IC_\mu(f) \le \IC_{\mu}(\pi) \le \sum\limits_{i=1}^k H(f_i(X))$.

\end{proof}
This lemma shows that $\IC$ cannot be used in the multi-party setting for any meaningful lower bounds on the communication complexity, since its value is always upper bounded by the entropies of the functions. 
Our goal is to get lower bounds tight in both $k$ and $n$. For this reason, we introduce a new information-theoretic quantity for the multi-party setting.

\section{The new measure: Public Information Cost}
\label{sec:pic}
We now introduce a new information theoretic quantity which can be used instead of $\IC$ in the multi-party setting.
The notion we define will be suitable for studying multi-party communication in a model which is only a slight restriction on the general asynchronous model, and which allows for private protocols. This means that while $\IC$ will be at most the entropies of the functions, our new notion remains a strong lower bound for communication. 
\begin{definition}
For any $k$-player protocol $\pi$ and any input distribution $\mu$, we define the \textit{public information cost} of $\pi$:
$$\PIC_\mu(\pi) = \sum\limits_{i=1}^k I(X_{-i} ; \Pi_i R_{-i} \mid X_i R_i R^p).$$
\end{definition}

The difference between the definition of $\PIC$ and that of $\IC$ is the presence of the other parties' private randomness, $R_{-i}$, in the formula. Thus, if $\pi$ is a protocol using only public randomness, then for any input distribution $\mu$, $\PIC_\mu(\pi) = \IC_\mu(\pi)$, and hence the name ``public information cost''.

Informally speaking, the public information cost measures both the information about the inputs learned by the players and the information that is hidden by the use of private coins. $\PIC$ can be decomposed, using the chain rule, into two terms, making explicit the contribution of the internal information cost and that of the private 
randomness of the players.

\begin{proposition}\label{p2}
For any $k$-player protocol $\pi$ and any input distribution $\mu$,
$$\PIC_\mu(\pi) = \IC_\mu(\pi) + \sum\limits_{i=1}^k I(R_{-i} ; X_{-i} | X_i \Pi_i R_i R^p).$$
\end{proposition}
A possible intuitive meaning of the second term could be the following.
 At the end of the protocol, player $i$ knows its input $X_i$, its private coins $R_i$, the public coins $R^p$ and its transcript $\Pi_i$. Suppose that the private randomness $R_{-i}$ of the other players is now revealed to player $i$. This brings to that player some new information, quantified by $I(R_{-i} ; X_{-i} | X_i \Pi_i R_iR^p)$, about the inputs $X_{-i}$ of the other players.

We also define the public information complexity of a function, given error probability $\epsilon$.
In the sequel, when clear from the context, we sometimes omit  $\epsilon$.
\begin{definition}
For any function $f$, $\epsilon \geq 0$,  and any input distribution $\mu$, we define the quantity
$$\PIC^{\epsilon}_\mu(f) =  \inf\limits_{\pi~\text{ $\epsilon$-computing }f} \PIC_\mu(\pi)~.$$
\end{definition}
\begin{definition}
For any $f$, we define the quantity
$$\PIC^{\epsilon}(f) =  \inf\limits_{\pi~\text{$\epsilon$- computing }f} \sup\limits_{\mu}~\PIC_\mu(\pi)~.$$
\end{definition}
The public information cost is a lower bound on the communication complexity.

\begin{proposition}\label{p3}
For any protocol $\pi$ and input distribution $\mu$, $\CC(\pi) \ge \PIC_\mu(\pi)$. 
Thus, for any function $f$, $\CC(f) \ge \PIC(f)$.
\end{proposition}
\begin{proof}
\begin{align*}
\PIC_\mu(\pi) &= \sum\limits_{i=1}^k I(X_{-i} ; R_{-i} \mid X_i R_i R^p) + I(X_{-i} ; \Pi_i \mid X_i R R^p)\text{~~~(by the chain rule)}\\
&=\sum\limits_{i=1}^k I(X_{-i} ; \Pi_i \mid X_i R R^p)\text{~~~(since the first term is $0$ by Proposition~\ref{indp})}\\
&=\sum\limits_{i=1}^k H(\Pi_i \mid X_i R R^p)\text{~~~(since $H(\Pi_i \mid X R R^p)=0$ )}\\
&\le \sum\limits_{i=1}^k H(\Pi_i)\text{~~~(by Proposition~\ref{condi})}~.
\end{align*}
Using Proposition~\ref{H b}, for each $i$, $H(\Pi_i)$ is upper bounded by the expected size of $\Pi_i$. As the expected size of $\Pi$ is upper bounded by the sum over $i$ of the expected size of $\Pi_i$, we get $\CC(\pi) \ge \PIC_\mu(\pi)$.

\end{proof}
In fact, as we show below, the public information cost of a function is equal to its  information cost (IC) in a setting where only public randomness is allowed. The role of private coins in communication protocols has been studied for example in \cite{BG,BBKLS,Koz}. In the next section we will see that the difference between the public information cost and the information cost is related to the private coins used during the protocol.

\begin{theorem}
\label{th:protocol_min_on_pub}
For any protocol $\pi$  there exists a public coin protocol $\pi'$ such that for all input distributions 
$\mu$, $\PIC_\mu(\pi') = \PIC_\mu(\pi)$. If $\pi$ is oblivious then so is $\pi'$.
\end{theorem}
\begin{proof}
Given an an arbitrary protocol $\pi$, we build a public coin protocol $\pi'$ as follows.\\
Let $R'^{p}$ denote the public random tape of $\pi'$. We consider $R'^{p}$  as being composed of $k+1$ parts, one to be used
as the public random tape of $\pi$, and the $k$ other parts as the $k$ private random tapes, in $\pi$,  of the $k$ players.
This can be done by interleaving the $k+1$ tapes bit by bit on $R'^{p}$. We denote the $k+1$ resulting tapes as
$R^{p}$ and $R_i$, $1 \leq i \leq k$. Protocol $\pi'$ is then defined as running protocol $\pi$, when the players use the
corresponding ``public  random tape''  $R^{p}$ and ``private random tapes'' $R_i$, as defined for $\pi$.  
Observe that the transcripts of $\pi'$ and $\pi'$ are therefore identical. Observe  also that if $\pi$ is oblivious, so is $\pi'$. Let $R$ denote $(R_i)$.
We have, 
\begin{align*}
\PIC_\mu(\pi') 
&= \sum\limits_{i=1}^k I(X_{-i} ; \Pi'_i \mid X_i  R'^p)\\
&= \sum\limits_{i=1}^k I(X_{-i} ; \Pi'_i \mid X_i  R R^p) \\
&= \sum\limits_{i=1}^k [I(X_{-i} ; \Pi'_i R_{-i} \mid X_i R_i R^p) -  I(X_{-i} ; R_{-i} \mid X_i R_i R^p)] \text{~~(chain rule)}\\
&= \sum\limits_{i=1}^k I(X_{-i} ; \Pi'_i R_{-i} \mid X_i  R_i R^p)\text{~~(since the second term equals $0$)}\\
&= \sum\limits_{i=1}^k I(X_{-i} ; \Pi_i R_{-i} \mid X_i  R_i R^p)\text{~~(since the transcripts of $\pi'$ and of $\pi$ are identical)}\\
&= \PIC_\mu(\pi)~.
\end{align*}
\end{proof}

The next theorem is a direct consequence of Theorem~\ref{th:protocol_min_on_pub}.
\begin{theorem}\label{min on pub}
For any function $f$ and input distribution $\mu$,
$$\PIC_\mu(f) =  \inf\limits_{\pi\text{ computing }f\text{, using only public coins}} \IC_\mu(\pi)$$
and
$$\PIC(f) =  \inf\limits_{\pi\text{ computing }f\text{, using only public coins}} \sup\limits_{\mu}~\IC_\mu(\pi)~.$$
\end{theorem}

The following property of the public information cost will be useful for zero-error protocols. 
\begin{proposition}\label{min on det}
For any function $f$, for any input distribution $\mu$, $\PIC^0_\mu(f) = \IC_\mu^{\text{det}}(f)$ where
$$\IC_\mu^{\text{det}}(f) =  \inf\limits_{\pi\text{ deterministic protocol computing }f} \IC_\mu(\pi)~.$$
\end{proposition}
\begin{proof}
Let $\delta > 0$ be arbitrary.  To prove the claim we show that there exists a deterministic protocol 
computing $f$,  $\pi^0$, such that $\IC_\mu(\pi^0) \leq \PIC^0_\mu(f) + \delta$.

Let $\pi$ be a zero-error protocol for $f$ such that ${\PIC_\mu(\pi) \le \PIC^0_\mu(f) + \frac{\delta}{2}}$. 
By Theorem \ref{min on pub}, one can assume that $\pi$ has no private randomness.
\begin{align*}
\IC_\mu(\pi) &= \sum\limits_{i=1}^k I(X_{-i} ; \Pi_i \mid X_i R^p) \\
&= \sum\limits_{i=1}^k \Exp\limits_{r} \left[I(X_{-i} ; \Pi_i \mid X_i, R^p=r)\right] \\
&= \Exp\limits_{r} \left[\sum\limits_{i=1}^k I(X_{-i} ; \Pi_i \mid X_i, R^p=r)\right]~.
\end{align*}
Letting $t(r) = \sum\limits_{i=1}^k I(X_{-i} ; \Pi_i \mid X_i, R^p=r)$, it holds that $\IC_\mu(\pi) = \Exp\limits_r \left[ t(r) \right]$.
Let $r_0$ be a value of the public random tape such that $t(r_0) \le \IC_\mu(\pi) + \frac{\delta}{2}$
and define $\pi^0$ to be  the  protocol operating like $\pi$ when  the random tape is $r_0$. 
Note that $\pi^0$ is a deterministic (zero-error) protocol computing $f$.
\begin{align*}
\PIC_\mu(\pi^0) & = \IC_\mu(\pi^0) \\
&= \sum\limits_{i=1}^k I(X_{-i} ; \Pi^0_i \mid X_i) \\
&= \sum\limits_{i=1}^k I(X_{-i} ; \Pi_i \mid X_i R = r_0) \\
&= t(r_0) \\
&\le \IC_\mu(\pi) + \frac{\delta}{2} \\
& \leq \PIC_\mu(\pi) + \frac{\delta}{2} \\
&\le \PIC^0_\mu(f) + \delta~.
\end{align*}
$\delta$ being arbitrary, this concludes the proof.
\end{proof}

We now observe that 
$\PIC$ and $\IC$ are strictly different even in the two-party case. 
We prove below that for the $\AND$ function, the public information cost is $\log (3)\simeq 1.58$, while, as shown in \cite{BGPW}, $\IC^0(\AND) \simeq 1.49 $. This implies that the protocol that achieves the optimal information cost for $\AND$ must use private coins. We remark  that in \cite{BGPW} it is shown that the external information cost of $\AND$, that we do not consider here, is $\log(3)$.
\begin{proposition}
For two players, $\PIC^0(\AND) = \log_23\simeq 1.58$.
\end{proposition}
\begin{proof}
In this proof we denote $\PIC^0(\cdot)$ by simply $\PIC(\cdot)$. 
We call a protocol $\pi'$ symmetric to $\pi$ (and an input distribution $\mu'$ symmetric to $\mu$)
when the roles of Alice and Bob (or of the inputs $X$ and $Y$) are flipped.

We first  prove  that there exists a  protocol $\pi^*$ for $\AND$ such that 
$\sup_\mu\PIC_\mu(\pi^*)= \inf_\pi \sup_\mu\PIC_\mu(\pi)$, where the infimum is over all
 protocols $\pi$ computing $\AND$. 
 To this end we  now prove that for any protocol $\pi$ for $\AND$ it holds that  $\sup_\mu\PIC_\mu(\pi^*) \leq \sup_\mu\PIC_\mu(\pi)$,
 where $\pi^*$ is a protocol for $\AND$ that we define below.
 Consider an arbitrary protocol $\pi$. By Proposition~\ref{min on det} we can assume w.l.o.g.
 that $\pi$ is deterministic. 
Since $\pi$ computes $\AND$ there must be a non-constant bit sent in $\pi$. Assume 
w.l.o.g. that the first player to send  a non-constant bit 
in $\pi$ is Alice (having input $X$).\footnote{We can assume this w.l.o.g. because for any protocol $\pi$, 
$\sup_\mu\PIC_\mu(\pi)  = \sup_\mu\PIC_\mu(\pi')$, where $\pi'$ is the protocol symmetric to $\pi$.
This is because any input distribution $\mu$ has a symmetric one, $\mu'$.}
Since $\pi$ is deterministic, this first non-constant  bit is either Alice's input
bit, $x$, or, $1-x$.
Since Alice must  compute the value of $\AND$, we also have  that $H(\AND(X,Y) | X\Pi_{A})=0$, where
$\Pi_A$ is the transcript of the messages received by  Alice.
Consider now the protocol $\pi^*$ defined as follows: Alice sends her input bit $x$ to Bob; Bob, who can now
 compute  $\AND(X,Y)$, sends to Alice that value.
 
For any input distribution $\mu$ we have
\begin{align*}
\PIC(\pi^*) & =  \IC_\mu(\pi^*) \\
& = I(X; \Pi^*_B \mid  Y) + I(Y; \Pi^*_A  \mid X) \\
&=  H(X \mid Y) + (H(Y \mid X) - H( Y \mid X ~\AND(X,Y))) \\
&=  I(X ; \Pi_B \mid Y) + (H(Y \mid X) - H( Y \mid X ~\AND(X,Y))) \\
& \leq I(X ; \Pi_B \mid Y) + (H(Y \mid X) - H( Y \mid X \Pi_A)) \text{~~~(because $H(\AND(X,Y) | X\Pi_{A})=0$) } \\
&=  I(X ; \Pi_B \mid Y) + I ( Y ; \Pi_A \mid X) \\
&= \PIC(\pi)~.
\end{align*}
 
 It follows that for any $\pi$ computing $\AND$, it holds that  $\sup_\mu\PIC_\mu(\pi^*) \leq \sup_\mu\PIC_\mu(\pi)$, 
 and hence $\sup_\mu\PIC_\mu(\pi^*)= \inf_\pi \sup_\mu\PIC_\mu(\pi)$.

To finalize the proof we now show that $ \sup_{\mu}\PIC_\mu(\pi^*) =\PIC_{\mu^*}(\pi^*)$
 for $\mu^*$ defined as follows: $X$ and $Y$ are independent; 
 $X \sim \Ber(\frac{1}{3},\frac{2}{3})$;  $Y \sim \Ber(\frac{1}{2},\frac{1}{2})$.

Consider an arbitrary input distribution $\mu$. Let $\alpha$ and $\beta$ be such that $\Pr_\mu[X=0]=\alpha$
and $\Pr_\mu[Y=0]=\beta$. Observe that  $X$ and $Y$ are  not necessarily independent.
We have
\begin{align*}
\PIC_\mu(\pi^*) &= I_\mu(X;\Pi^* \mid Y) + I_\mu(Y ;\Pi^* \mid X)\\
&= H_\mu(X \mid Y) + \left[\alpha \cdot I_\mu(Y ;\Pi^* \mid X=0) + (1-\alpha)\cdot  I_\mu(Y ;\Pi^* \mid X=1)\right]~,
\end{align*}
where the second equality follows from $H_\mu(X \mid Y \Pi^*)$=0, as the transcript  of $\pi^*$ fully determines $X$.
Observe that when $X=0$, Alice doesn't learn  from $\Pi^*$ anything about $Y$, while when $X=1$, Alice learns from $\Pi^*$ the value of $Y$.
Thus
$$\PIC_\mu(\pi^*) = H_\mu(X \mid Y) + (1-\alpha)\cdot H_\mu(Y \mid X=1)~.$$

Now define another input distribution $\mu'$ such that: $X$ and $Y$ are independent; 
$X \sim \Ber(\alpha,1-\alpha)$;  $\Pr[Y=1] = \Pr_{\mu}[Y=1 \mid X=1]$.
Note that 
$H_{\mu'}(X \mid Y) = H_{\mu'}(X) = H_{\mu}(X)$ and that
$H_{\mu'}(Y \mid X=1) = H_{\mu'}(Y) = H_\mu(Y \mid X=1)$.
We thus have that
\begin{align*}
\PIC_{\mu'}(\pi^*) &= I_{\mu'}(X;\Pi^* \mid Y) + I_{\mu'}(Y ;\Pi^* \mid X) \\
&= H_{\mu'}(X \mid Y) + (1-\alpha)\cdot H_{\mu'}(Y \mid X=1) \\
&= H_\mu(X)  + (1-\alpha)\cdot H_\mu(Y \mid X=1) \\
& \geq  H_\mu(X \mid Y )  + (1-\alpha)\cdot H_\mu(Y \mid X=1)\\
& = \PIC_\mu(\pi^*)~.
\end{align*}

Therefore, to find $\sup_\mu\PIC_\mu(\pi^*)$ we can consider only input 
distributions $\mu'$ such that $X$ and $Y$ are independent. For such $\mu'$ we define $\alpha'$ and $\beta'$ such that $X \sim \Ber(\alpha',1-\alpha')$ and
$ Y \sim \Ber(\beta',1-\beta')$. We have
$$\PIC_{\mu'}(\pi^*) = H_{\mu'}(X) + (1-\alpha')H_{\mu'}(Y)~.$$

  Thus, for any $\alpha'$, $\PIC_{\mu'}(\pi^*)$ is maximized when $H_{\mu'}(Y) = 1$, i.e.,
  when $\beta'=\frac{1}{2}$. In that case we have $\PIC_{\mu'}(\pi^*) = H_{\mu'}(X) + (1-\alpha')$.
  Thus,  to find $\sup_\mu\PIC_\mu(\pi^*)$,  we study the function
$f:[0,1]\rightarrow\mathbb{R}$, defined as 
$f(\alpha')=-\alpha'\log(\alpha') + (\alpha' - 1)\log(1-\alpha') + 1 - \alpha'$. \footnote{We denote here by $\log$ the logarithm base $2$.}

Now,
$f$ is continuous on $[0,1]$ and differentiable on $(0,1)$. For $ 0 < \alpha\ <1 $, we have:\\
(1) $f'(\alpha') = -\log(\alpha') - 1 + \log(1-\alpha') + 1 - 1 = \log(\frac{1}{\alpha'} - 1) - 1$; (2)
$f'$ is continuous and decreasing on $(0,1)$;  and (3) $f'$ admits the unique root $\frac{1}{3}$. 
Thus, $f$ is  maximized for $\alpha' = \frac{1}{3}$, its  maximum value being $f(\frac{1}{3}) = \log(3)$.

We thus have that $\sup_\mu\PIC_\mu(\pi^*)=\PIC_{\mu^*}(\pi^*)$ for
 $\mu^*$ defined as follows: $X$ and $Y$ are independent; 
 $X \sim \Ber(\frac{1}{3},\frac{2}{3})$;  $Y \sim \Ber(\frac{1}{2},\frac{1}{2})$,  that 
 $\PIC_{\mu^*}(\pi^*)=\log(3)$, and that 
 $\PIC(\AND) = \log(3)\simeq 1.58$.
\end{proof}

\section{Private computation, randomness, and PIC}
\label{sec:privacy}
We have seen that the public information cost of a function is equal to the information cost of the function when we only consider public coin protocols, and that in order to decrease the information cost even further, the players must use private randomness. We will see now that the difference between the public information cost of a protocol and its information cost can provide a lower bound on the amount of private randomness the players use during the protocol. The entropy of the transcript of the protocol, conditioned on the inputs and the public coins, is defined as $H(\Pi \mid X R^p)$. Once the input and the public coins are fixed, the entropy of the transcript of the protocol comes solely from the private randomness. Thus the entropy of the transcript of the protocol provides a lower bound on the entropy of the private randomness used by the players.

\begin{theorem}\label{bound partial pri}
Let $f = (f_i)$ be a family of functions of $k$ variables. Let $\pi$ be a protocol for $f$. For any input distribution $\mu$, it holds:
$$H_\mu(\Pi \mid X R^p) \ge \frac{\PIC_\mu(\pi) - \IC_\mu(\pi)}{k}~.$$
Thus, running a protocol for $f$ with information cost $I_\mu$ requires entropy $$H_\mu(\Pi \mid X R^p) \ge \frac{\PIC_\mu(f) - I_\mu}{k}~.$$
\end{theorem}

\begin{proof}
We assume in what follows the input distribution $\mu$ without explicitly denoting it.

Define $Q_i$ as 
$$
Q_i=I( X_{-i}  ; R_{-i} \mid X_i R_i  \Pi_i R^p )~.
$$

By Proposition~\ref{p2} we have, 
\begin{align*}
\PIC(\pi) & = \IC(\pi) + \sum\limits_{i=1}^k I( X_{-i} ; R_{-i} \mid X_i R_i  R^p \Pi_i) \\
&= \IC(\pi) + \sum\limits_{i=1}^k Q_i~.
\end{align*}

Now,
\begin{align*}
Q_i  &= I(X_{-i} ; R_{-i}   \mid X_i R_i   \Pi_i R^p) \\
&= I(X_{-i}\Pi_i ; R_{-i} \mid X_i R_i R^p) - I(\Pi_i; R_{-i} \mid X_i R_i R^p) \text{~~~(chain rule)} \\
& \leq I(X_{-i}\Pi_i ; R_{-i} \mid X_i R_i R^p)  \\
& = I(X_{-i}; R_{-i} \mid X_i R_i R^p) + I(\Pi_i; R_{-i} \mid X_i R_i X_{-i} R^p) \text{~~~(chain rule)} \\
&= I(\Pi_i ; R_{-i} \mid X R_i R^p) \\
& = H(\Pi_i \mid X R_i R^p)  \\
& \leq H(\Pi \mid X R^p)~.
\end{align*}

Thus,
\begin{align*}
\PIC(\pi) &\le \IC(\pi) + k \cdot H(\Pi \mid  X R^p)~.\\
\end{align*}
\end{proof}

Using Lemma $\ref{p1}$, we can give a lower bound on the randomness required to run a private protocol.
\begin{corollary}\label{bound pri}
Let $f = (f_i)$ be a family of functions of $k$ variables. Let $\pi$ be a $k$-party private protocol for $f$. For any distribution $\mu$ on inputs,
$$H_\mu(\Pi \mid X R^p) \ge \frac{1}{k}\cdot \left(\PIC_\mu(f) - \sum\limits_{i=1}^k H_\mu(f_i)\right)~.$$
\end{corollary}

\section{Tight lower bounds for the parity function $\XOR$}
\label{sec:xor}
We now show how one can indeed use $\PIC$ to study multi-party communication protocols and to prove tight bounds. We study one of the canonical problems for zero-error multi-party computation, the parity function. The $k$-party parity problem with $n$-bit inputs $\XOR_k^n$ is defined as follows. 
Each player $i$ receives $n$ bits $(x_i^p)_{p \in \Ni{1}{n}}$ and Player 1 has to output the bitwise XOR of the inputs  $\left( \bigoplus \limits_{i=1}^k x_i^1, \bigoplus\limits_{i=1}^k x_i^2,\ldots, \bigoplus\limits_{i=1}^k x_i^n \right) $. We 
give a lower bound on $\XOR_k^n$ and then use it to prove tight lower bounds on the randomness complexity 
of private computations of $\XOR_k^n$.

There is a simple private protocol for $\XOR_k^n$ that uses $n$ bits of private randomness. Player 1 uses a private random $n$-bit string $r$ and sends to Player 2 the string $x_1 \oplus r$. Then, Player 2 computes the bit-wise parity of its input with that message and sends $x_2 \oplus x_1 \oplus r$ to Player 3. The players continue until Player 1 receives back the message $x_k\oplus \ldots \oplus x_1 \oplus r$. Player 1 then takes the bit-wise parity of this message with the private string $r$ to compute the value of the parity function. It is easy to see that this protocol has information cost equal to $n$, since Player 1 just learns the value of the function and all other players learn nothing. We thus see that information cost (IC) cannot provide here lower bounds that scale with $k$.

We note that we prove our lower bound for $\XOR_k^n$  for a wider class of protocols, where we allow the 
player outputting $\oplus_{i=1}^k x_i^p$ to be different for each coordinate $p$ and where the identity of that player may  depend on the input.  On the other hand, we prove our lower bound for the restricted class of $0$-error oblivious protocols. We now prove a tight lower bound   
of $\Omega(nk)$ on the $\PIC$ of  $\XOR_k^n$  (for $0$-error oblivious protocols) 
which can then 
be used to derive other lower bounds for protocols for $\XOR_k^n$. 

For the purpose of the proof we 
 define a (natural) full order on the messages of an oblivious protocol. The order is defined as follows. We define an ordered series of {\em lots of messages}. In each lot there is at most  one message on any of the $k(k-1)$ directional links. The order of the messages is defined by the order of the lots, and within
  each lot, the messages are ordered by, say, the 
 lexicographical order of the links on which they are sent. The messages are assigned to lots as follows: The first lot consists of all messages sent by all the players in their first respective local round. The messages assigned to lot $s\geq1$ are defined inductively after lots $s'<s$ have been defined. To define the messages of lot $s>1$, we proceed 
 as follows for each player $i$: run the  protocol $\pi$, and whenever player $i$ is waiting for a message, extract a message from the already defined lots (lots $s'<s$), if  such message is assigned to one of them. Continue until a needed message is not available (i.e., the 
 protocol ``gets stuck''), or after player $i$ sends, according to the protocol, a  message not already assigned to a lot $s'<s$. 
 In the latter case, assign to lot $s$  all the  messages sent by player $i$ in 
the same  local round (i.e.,  for any player $i$ and local round $r$, all messages sent by player $i$ in local
 round $r$ are in the same lot).  
 
 To see that all the messages of the protocol are assigned to lots, build the following graph where each node
is identified by a pair $(i,r)$, for a player $i$ and local round $r$ of player $i$.   There is a directed edge
from any node $(i,r')$ to node $(i,r)$, if  $r'<r$  and there is at least one message sent by player $i$ in round $r'$.
Further, there is a directed edge from node $(j,r')$ to node $(i,r)$ if there is an integer $\ell$ such that 
 the $\ell$'th message from player $j$ to player $i$ is sent by player $j$ in its local round $r'$ and read by player $i$
 in its local round $r$. Observe that a node $(i,r)$
 is not on a directed cycle if and only if,  when the protocol is run, player $i$ reaches the sending-of-messages phase of its local round $r$.
  Define a partial order on the nodes which are not on a directed cycle, according to the orientation of the edges. 
  We define the ``level'' of a node to be the length of the  longest directed path leading to it.  Observe that by induction on this level, the messages sent by player $i$ in
 local round $r$, where node $(i,r)$ is of level $s$, are assigned to lot number $s$.

  Observe   that the enumeration of the messages as defined above respects the 
  intuitive ``temporal causality'' of the messages of the protocol. More formally, the following two properties hold for
  the above defined order: (1) the relative order of  the local rounds of  two messages that are both sent from player, say, $i$, to player, say, $j$,  
  is the same as the relative order of these messages according to the global order, and 
 (2) the value 
  of a message number $\ell$ (in the global order) sent from player $i$ is fully determined by the input to player
   $i$ and the values of the messages with indices less than $\ell$ that are received by player $i$.
  
 Denote by $(\M{i}{l}{s}\,)_{l\ge 0}$ the ordered sequence  of all messages sent by player $i$ in the protocol $\pi$, ordered according to the enumeration defined above.
 Similarly, denote by $(\M{i}{l}{r}\,)_{l\ge 0}$ the ordered sequence  of messages received by player $i$. Denote by $j(i,l)$ the player receiving message $\M{i}{l}{s}$, and
 by $l'(i,l)$ the integer such that the random variable  $\M{j(i,l)}{l'(i,l)}{r}$ and the random variable $\M{i}{l}{s}$ represent the same message. Observe that since we consider here an 
 oblivious protocol, the functions $j(i,l)$ and $l'(i,l)$ are well defined.
 For any $l_0$, let $\T{i}{l_0}{s} $ be the random variable representing the so-far history
 of player $i$, i.e., all the messages to and from player $i$ which appear before message
  $\M{i}{l_0}{s}$
 in the enumeration of messages defined above.
In a similar way, define $\T{i}{l_0}{r} $ to be the random variable representing the so-far history of the messages to and from player $i$  which appear before message 
 $\M{i}{l_0}{r}$.

\begin{theorem}\label{thm xor}
For oblivious protocols, 
$\PIC^0_\mu(\XOR_k^n) \ge n(k-1)$ where $\mu$ is the uniform input distribution.
\end{theorem}

\begin{proof}
Throughout the proof we consider the uniform input distribution $\mu$ without explicitly stating it. 
Since we are looking at $0$-error protocols, the public information cost is equal to the information cost of deterministic protocols. 
Let $\pi$ be a $0$-error deterministic protocol for $\XOR_k^n$ for $k$ players and $n$-bit input  per player. 

We first prove that 
\begin{equation}
\label{eq:PIC>Spy}
\PIC^0(\pi) \ge \sum\limits_{i=1}^k I(X_i ; \overleftrightarrow{_{~}\Pi_i})~. 
\end{equation}

\noindent Intuitively, this means that $\PIC$ is at least the sum over the players $i$ of the amount of information that player $i$ leaks about its input to 
some entity that has access to all messages to and from player $i$.

\noindent Since $\pi$ is a deterministic protocol we have 
$\PIC^0(\pi) =\sum\limits_{j=1}^k I(X_{-j}; \Pi_j \mid X_j)$. We will therefore show
that $\sum\limits_{j=1}^k I(X_{-j}; \Pi_j \mid X_j) \geq \sum\limits_{i=1}^k I(X_i ; \overleftrightarrow{_{~}\Pi_i})$.

Using the chain rule, 
we decompose  $\sum\limits_{j=1}^k I(X_{-j}; \Pi_j \mid X_j)$ into a sum over all 
messages received in the protocol:  
\begin{align*}
\sum\limits_{j=1}^k I(X_{-j}; \Pi_j \mid X_j) &= 
 \sum\limits_{j=1}^{k}\sum\limits_{\ell' \geq 0} I(X_{-j} ; \M{j}{\ell'}{r} \mid \M{j}{0}{r}\ldots \M{j}{\ell'-1}{r}X_{j} ) \\
&= \sum\limits_{j=1}^{k}\sum\limits_{\ell' \geq 0} I(X_{-j} ; \M{j}{\ell'}{r} \mid \T{j}{\ell'}{r} X_{j} )~.
\end{align*}

We now consider each message from the point of view of the receiver rather than that of the sender.  Recall that each message in the protocol is represented by two random
variables: for any $i$ and $l$ the two random variables 
 $\M{i}{l}{s}$ and $\M{j(i,l)}{l'(i,l)}{r}$ represent the same message. Thus, we can rearrange the last summation, using $j$ as a shorthand for $j(i,l)$  and
 $l'$ as a shorthand of $l'(i,l)$, and get 
\begin{align*}
\sum\limits_{j=1}^k I(X_{-j}; \Pi_j \mid X_j) &= 
\sum\limits_{i=1}^k\sum\limits_{\ell \geq 0} I(X_{-{j}} ; \M{i}{\ell}{s} \mid \T{j}{\ell'}{r} X_{j})~.
\end{align*}

Note that using the chain rule, we have for all $i\in\Ni{1}{k}$,
\begin{align*}
 I(X_i ; \overleftrightarrow{_{~}\Pi_i}) &= \sum\limits_\ell I(X_i ; \M{i}{\ell}{s} \mid \T{i}{\ell}{s}) +  \sum\limits_\ell I(X_i ; \M{i}{\ell}{r} \mid \T{i}{\ell}{r})\\ 
 &= \sum\limits_\ell I(X_i ; \M{i}{\ell}{s} \mid \T{i}{\ell}{s})~,
\end{align*}
where we used the fact that every term of the second sum is $0$. This is true using Proposition~\ref{indp}, which can we used since, for 
any $\ell$, conditioned on $\T{i}{\ell}{r}$, $X_i$ is independent of the variable $\M{i}{\ell}{r}$ (intuitively: when the input distribution is a product distribution, the incoming messages to a player do not carry any information on the input of that player).

Therefore, our objective  now is to show that for any message $\M{i}{\ell}{s}$, 
\begin{equation}
\label{eq:to_prove}
I(X_{-{j}} ; \M{i}{\ell}{s} \mid \T{j}{\ell'}{r} X_{j})    \geq  I(X_i ; \M{i}{\ell}{s} \mid \T{i}{\ell}{s} )~.
\end{equation}

Since $\M{i}{l}{s}$ is determined by $(X_i,\T{i}{l}{s})$, we have that $H(\M{i}{l}{s} \mid X_i \T{i}{l}{s}) = 0$, and $I(X_i ; \M{i}{l}{s} \mid \T{i}{l}{s}) = H(\M{i}{l}{s} \mid \T{i}{l}{s})$. Similarly (as $X_i$ is trivially a function of $X_{-j}$),
$I(X_{-{j}} ;\M{i}{l}{s} \mid \T{j}{l'}{r} X_{j}) = H(\M{i}{l}{s} \mid \T{j}{l'}{r} X_{j})$.

Thus,
\begin{align*}
I(X_i ; \M{i}{\ell}{s} \mid \T{i}{\ell}{s}) &\le I(X_{-{j}} ; \M{i}{\ell}{s} \mid \T{j}{\ell'}{r} X_{j})\\
&\Updownarrow\\
H(\M{i}{\ell}{s} \mid \T{i}{\ell}{s}) &\le H(\M{i}{\ell}{s} \mid \T{j}{\ell'}{r} X_{j})\\
&\Updownarrow\\
I(\M{i}{\ell}{s} ; \T{i}{\ell}{s}) &\ge I(\M{i}{\ell}{s} ; \T{j}{\ell'}{r} X_{j})~.
\end{align*}

The last inequality holds if $I(\M{i}{\ell}{s} ; \T{i}{\ell}{s})  = I(\M{i}{\ell}{s} ; \T{i}{\ell}{s}\T{j}{\ell'}{r} X_{j})$,
which itself holds if
\begin{equation}
\label{eq:strange_one}
I(\M{i}{\ell}{s} ; \T{j}{\ell'}{r} X_{j} \mid \T{i}{\ell}{s})=0~.
\end{equation} 

Observe that given $\T{i}{\ell}{s}$, $\M{i}{\ell}{s}$ is fixed by $X_i$, and therefore by the data processing inequality 
$$
I(X_i ; \T{j}{\ell'}{r} X_j \mid \T{i}{\ell}{s}) \geq I(\M{i}{\ell}{s}; \T{j}{\ell'}{r} X_j \mid \T{i}{\ell}{s})~,
$$
and thus Equality~\eqref{eq:strange_one} holds  if 
\begin{equation}
\label{eq:important_one}
I(X_i ; \T{j}{\ell'}{r} X_j \mid \T{i}{\ell}{s})= 0~.
\end{equation}

Observe now that the ordering of the messages that we defined implies that $\T{j}{\ell'}{r}$ (which is the same message as
$\M{i}{\ell}{s}$) is determined  by $(X_{-i},\T{i}{\ell}{s})$.
Furthermore, $X_j$ is trivially determined by $X_{-i}$. Using the data processing inequality we thus have
$$
I(X_i ; \T{j}{\ell'}{r} X_j \mid \T{i}{\ell}{s}) \leq I(X_i ; X_{-i} \T{i}{\ell}{s} \mid  \T{i}{\ell}{s})=I(X_i ; X_{-i}  \mid  \T{i}{\ell}{s})~.
$$

Thus Equality~\eqref{eq:important_one} holds if $I(X_i ; X_{-i}  \mid  \T{i}{\ell}{s})=0$.
To prove the latter,
denote by $(B^d)_{d>0}$ all the messages in $\T{i}{\ell}{s}$, ordered by local rounds of player $i$, 
and inside each round having first the messages sent by 
player $i$, ordered by the index of the recipient, and  then the messages received by player $i$, ordered by the 
index of the sender. For convenience of notation we also define the message $B^0$, which is the ``empty message'' at the beginning
 of $\T{i}{\ell}{s}$.
 
 We  prove, by induction on $d$, that for any $d \geq 0$, ${I(X_i ; X_{-i} \mid B^0B^1\dots B^d)=0}$.
 For the base of the induction ($d=0$) we have $I(X_i ; X_{-i} \mid B^0)=I(X_i ; X_{-i} )=0$, since $X$ is distributed according
  to $\mu$.
   
 By the induction hypothesis, for some $d\geq0$, ${I(X_i;X_{-i}\mid  B^0\dots B^d)=0}$. 
 If the message $B^{d+1}$ is {\em sent} by player $i$, then $B^{d+1}$ is 
 a function of $X_i$ and $B^0\dots B^{d}$, and thus
\begin{align*}
I(X_i;X_{-i}\mid B^0\dots B^{d+1}) &= 
H(X_{-i}\mid B^0\dots B^{d+1})~- H(X_{-i}\mid B^0\dots B^{d+1} X_i)\\
&\le H(X_{-i}\mid B^0\dots B^{d})~- H(X_{-i}\mid  B^0\dots B^{d} X_i)\\
&= I(X_i;X_{-i}\mid  B^0\dots B^{d})\\
&=0~.
\end{align*}
Similarly, if the message $B^{d+1}$ is {\em received} by player $i$, then $B^{d+1}$ is a
 function of $X_{-i}$ and $B^0\dots B^{d}$, and thus
\begin{align*}
I(X_i;X_{-i}\mid  B^0\dots B^{d+1}) &= 
H(X_i\mid  B^0\dots B^{d+1})~- H(X_i\mid  B^0\dots B^{d+1} X_{-i})\\
&\le H(X_i\mid  B^0\dots B^{d})~- H(X_i\mid  B^0\dots B^{d} X_{-i})\\
&= I(X_i;X_{-i}\mid  B^0\dots B^{d} )\\
&=0~.
\end{align*}

We thus have that  $I(X_i ; X_{-i}  \mid  \T{i}{\ell}{s})=0$
and  the proofs of Inequality~\eqref{eq:important_one} and of Equality~\eqref{eq:strange_one} are concluded. 
Inequality~\eqref{eq:to_prove} and  Inequality~\eqref{eq:PIC>Spy} then follow. 

\medskip 

To conclude the proof of the theorem we now show that 
\begin{equation}
\label{eq:bigI}
\sum\limits_{i=1}^k I(X_i ; \overleftrightarrow{_{~}\Pi_i}) \ge n(k-1)~.
\end{equation}

Let $x=(x_i^p) \in  \{0,1\}^{nk}$ be an arbitrary input, where $x_i$ is the $n$-bit input of player $i$.
 For any index  $1 \leq p \leq n$, any player $1 \leq q \leq k $, we consider the question whether
 $H(\bigoplus\limits_{i=1}^k x_i^p \mid  X_q=x_q, \Pi_q=\pi_q(x))=0$. Intuitively, if this is the case then
 player $q$ can output $\bigoplus\limits_{i=1}^k x_i^p$. Observe that since $\pi$ is a $0$-error protocol
 for $\XOR_k^n$, then for each index $p$ there is at least one player $q$ such that 
 $H(\bigoplus\limits_{i=1}^k x_i^p \mid  X_q=x_q, \Pi_q=\pi_q(x))=0$. Denote by $q^p(x)$ an arbitrary such
 player. 
For any player $i$, define $C_i(x) = \{p \mid q^p(x) \neq i\}$. Intuitively, when the input is $x$, then for each such coordinate, player $i$ has  to leak its own
input on that coordinate.  Let $c_i(x) = |C_i(x)|$.\\

We now show that
$\forall~i,~H(X_i \mid \overleftrightarrow{_{~}\Pi_i} = \overleftrightarrow{_{\!\;}\pi_i}(x)) \le n - c_i(x)$.
Assume towards a contradiction  that for some $i$, $H(X_i \mid \overleftrightarrow{_{~}\Pi_i} = \overleftrightarrow{_{\!\;}\pi_i}(x)) > n - c_i(x)$. This implies that the number of possible values for $X_i$ consistent with $\overleftrightarrow{_{~}\Pi_i} = \overleftrightarrow{_{\!\;}\pi_i}(x)$ is more than $2^{n-c_i(x)}$, and thus the number of coordinates of the input of the $i$-th player that are fixed by $ \overleftrightarrow{_{~}\Pi_i} = \overleftrightarrow{_{\!\;}\pi_i}(x)$ is strictly less than $c_i(x)$. 
In particular there exists an input $x'$ such that 
\begin{itemize}
\item $\overleftrightarrow{_{\!\;}\pi_i}(x') = \overleftrightarrow{_{\!\;}\pi_i}(x)$, and 
\item $\exists~p \in C_i(x)$ such that $x'^p_i \neq x_i^p$. 
\end{itemize}

Observe that we consider here {\em oblivious} multi-party protocols. Therefore,  $\overleftrightarrow{_{\!\;}\pi_i}(x') = \overleftrightarrow{_{\!\;}\pi_i}(x)$ implies that $\overleftrightarrow{_{\!\;}\pi_i}(x) = \overleftrightarrow{_{\!\;}\pi_i}(x'_i,x_{-i})$ (by considering  player $i$ as Alice, and
 all other players together as Bob, and using  arguments as those used  for a similar property for $2$-party protocols). 
As $q^p(x) \neq i$, this is a contradiction, since if $\overleftrightarrow{_{\!\;}\pi_i}(x) = \overleftrightarrow{_{\!\;}\pi_i}(x'_i,x_{-i})$ then 
(the output) $\bigoplus\limits_{i=1}^k x_i^p$ is not fixed by $X_q=x_q$  and   $\Pi_{q^p(x)}=\pi_{q^p(x)}(x)$,
contradicting  the definition of $q^p(x)$.

We now consider, for a given player $i$,  the quantity $\Exp\limits_x[c_i(x)]$.
For any given $x$ and any given player $i$ we  proved above that 
$H(X_i \mid \overleftrightarrow{_{~}\Pi_i} = \overleftrightarrow{_{\!\;}\pi_i}(x)) \le n - c_i(x)$.
Thus, for any player $i$ we have 
$H(X_i \mid \overleftrightarrow{_{~}\Pi_i}) \le \Exp\limits_x[n - c_i(x)] = n - \Exp\limits_x[c_i(x)]$. We get 
$I(X_i ; \overleftrightarrow{_{~}\Pi_i}) \ge \Exp\limits_x[c_i(x)]$.

Summing over all $i$, we get
$ \sum\limits_{i=1}^k I(X_i ; \overleftrightarrow{_{~}\Pi_i}) \ge \sum\limits_{i=1}^k\Exp\limits_x[c_i(x)] = \Exp\limits_x[\sum\limits_{i=1}^k c_i(x)]$ and since by  simple counting, for any $x$, it holds that   $\sum\limits_{i=1}^k c_i(x) = n(k-1)$, we get $\sum\limits_{i=1}^k I(X_i ; \overleftrightarrow{_{~}\Pi_i}) \ge \Exp\limits_x[n(k-1)] = n(k-1)$. This concludes the proof of 
Inequality~\eqref{eq:bigI}.

Inequality~\eqref{eq:PIC>Spy} together with Inequality~\eqref{eq:bigI} conclude the proof of the theorem.
\end{proof}

\begin{theorem}
The entropy in the private randomness of an oblivious  private protocol for $\XOR_k^n$ is at least $ \frac{k-2}{k} \cdot n$.
\end{theorem}
\begin{proof}
For $\XOR_k^n$, where one player outputs the parity for each coordinate, we have $\sum\limits_{i=1}^k H(f_i) = n$. 
Applying Corollary~\ref{bound pri}, we get:
$H(\Pi \mid X R^p) \ge \frac{k-2}{k} \cdot n$.

\end{proof}

We note that all private protocols considered in the literature are oblivious protocols. 
Observe  also that using Theorem~\ref{bound partial pri} one can also give a lower bound on the randomness needed by
protocols that are allowed to leak a given limited amount of  information about the inputs of the players.

\section{A direct sum for $\PIC$ ?}
\label{sec:direct-sum}

The \textit{direct sum} property is a fundamental question in complexity theory, and has been studied for many computation models. A direct sum theorem affirms that the amount of resources needed to perform $t$ independent tasks is at least the sum of the resources needed to perform each of the $t$ tasks. In this section we show that a certain direct sum property for $\PIC$ implies a  certain direct sum property for $\CC$. To this end, we prove a compression result by extending previous results~\cite{BBKLS,Pan} to the multi-party case. 
Note that information complexity ($\IC$) has a direct sum property in the multi-party case. For $\PIC$, it is easy to prove the following inequality.
\begin{theorem}
For any $k$-variable functions $f$ and $g$, for any distribution $\mu$ on  the inputs of $f$, for any distribution $\eta$ on  the inputs of $g$, it holds that $$\PIC_{\mu \times \eta}(f \times g) \le \PIC_{\mu}(f) + \PIC_{\eta}(g)~.$$
\end{theorem}
We use here the notation $f \times g$ to indicate  the task of computing $f$ with error $\epsilon$ and computing $g$ with error $\epsilon$ (as opposed  to computing the couple function $(f,g)$ with error $\epsilon$).
In order to understand whether the opposite inequality holds, i.e., whether a direct sum property holds for $\PIC$, we first need to study the problem of compressing communication.

\subsection{Relation between $\PIC$ and $\CC$: A compression result}

An important open question is how well  can we compress the communication cost of an interactive  protocol. Compression results have appeared in \cite{BBCR,BR,BBKLS,Pan,BMY}, while, on the other hand, \cite{GKR,GKR2,RS,FJKLLR,GKR3} focus on the hardness of compressing communication protocols. Here, we present a compression result with regards to the average-case communication complexity, distributional error, and the public information cost.

\begin{definition}
Given an input distribution $\mu$, a protocol is said to compute a function with distributional error $\epsilon$ if the probability, over the input and the randomness of the protocol, that the protocol fails is at most $\epsilon$.
\end{definition}

\begin{definition}
The \textit{average-case communication complexity} of a protocol $\pi$ with respect to the input distribution $\mu$, denoted $\ACC_\mu(\pi)$, is the expected number of bits that are transmitted in an execution of $\pi$ for inputs distributed according to $\mu$ and for uniform randomness.
\end{definition}

\begin{theorem}\label{compression pic}
Suppose there exists an oblivious protocol $\pi$ to compute a $k$-variable function $f$ over the distribution $\mu$ with distributional error probability $\epsilon$. Then for any fixed $\delta >0$ 
 there exists a public-coin protocol $\rho$ that computes $f$ over $\mu$ with distributional error $\epsilon + \delta$, and with average communication complexity
$$
\ACC_\mu(\rho) = \calO\left(k^2 \cdot \IC_\mu(\pi)\cdot\log\frac{k^2 \cdot \IC_\mu(\pi) \cdot \CC(\pi)}{\delta}\right).
$$	
\end{theorem}

The proof of the above theorem will follow from extending, to the case of $k > 2$ players, the compression result presented in \cite{BBKLS,Pan},
as stated below. Thus, the proof of Theorem~\ref{compression pic} follows  from Theorem~\ref{th:protocol_min_on_pub} and from Theorem~\ref{compression}. We remark that it is an interesting question whether the $k^2$ factor is necessary or whether it can be replaced by  smaller function of $k$.

\begin{theorem}\label{compression}
Suppose there exists an oblivious public-coin protocol $\pi$ to compute a $k$-variable function $f$ over the distribution $\mu$ with distributional error probability $\epsilon$. Then for any fixed $\delta >0$ there exists a public-coin protocol $\rho$ that computes $f$ over $\mu$ with distributional error $\epsilon + \delta$, and with average communication complexity
$$
\ACC_\mu(\rho) = \calO\left(k^2 \cdot \IC_\mu(\pi)\cdot\log\frac{k^2 \cdot \IC_\mu(\pi) \cdot \CC(\pi)}{\delta}\right).
$$
\end{theorem}
 
In the two-party compression scheme of \cite{BBKLS,Pan}, the two players, given their corresponding inputs, try to guess the transcript $\pi(x_1,x_2)$ of the protocol $\pi$. For this, player $1$ picks a candidate $t_1$  from the set $\Ima(\pi(x_1,\cdot))$ of possible transcripts  consistent with input $x_1$, while player $2$ picks a candidate $t_2$ from the set $\Ima(\pi(\cdot,x_2))$. The two players then communicate in order  to find the first bit on which $t_1$ and $t_2$ disagree. The general structure of protocols ensures that the common prefix of $t_1$ and $t_2$
(until the first bit of disagreement) is identical to the beginning of the correct transcript on inputs $x_1$ and $x_2$, i.e., identical to $\pi(x_1,x_2)$. Starting from this correct prefix, the players then pick new candidates for the transcript of the protocol $\pi(x_1,x_2)$, and so on, until they agree on the full transcript $\pi(x_1,x_2)$.
Clever choices of the candidates, along with an efficient technique to find the first bit which differs between the candidates, lead to a protocol with a small amount of communication.

In extending the proof in \cite{BBKLS,Pan} to the multi-party case new difficulties are encountered. The players can no longer try to guess the full transcript, as they have little information about the communication between the other players, and can only try to guess their partial transcripts, according to their own input. Then, in order to find the first disagreement  in the global transcript,  pairs of players need to find and communicate the place 
of the first disagreement between  their respective partial transcripts. 

For technical reasons, in this section we use the notation $\overleftrightarrow{_{\!\;}\Pi_i}$ not  as defined in 
Section~\ref{sec:model} to denote the 
 concatenation of $\Pi_i$ together with a similar string  $\overrightarrow{_{~}\Pi_i}$ of the messages sent by player $i$.
 Rather, we define $\overleftrightarrow{_{\!\;}\Pi_i}$ as a concatenation, local round of player $i$ by local round of player $i$, of, first, the messages
 sent by player $i$ and, then, the messages received by player $i$. Observe that since in this section we consider oblivious protocols there is
 a one-to-one correspondence between the transcripts of player $i$, $\overleftrightarrow{_{\!\;}\Pi_i}$, according to the two definitions.

Following~\cite{BBKLS}, in the definition of our protocol we will use a two-party 
``device'' 
as a black box, call it the \textit{lcp box} (for \textit{longest common prefix}), which can be used by two players $A$ and $B$ in the following way: $A$ inputs a string $x$, $B$ inputs a string $y$, and the box returns the first index $j$ such that $x_j \neq y_j$, if $x \neq y$, or returns  that $x=y$, otherwise. 
The conceptual device is assumed to operate with $0$ communication complexity.

This black box device can be efficiently simulated if we allow error:
\begin{lemma}[\cite{FRPU}]
\label{le:lcp_simulation}
For any $\epsilon >0$,
there exists  a randomized public coin protocol, such that on input two $n$-bits strings $x$ and $y$, it outputs the first index $j$ such that $x_j \neq y_j$ with probability at least $1-\epsilon$, if such $j$ exists, and otherwise outputs that the two strings are equal. The communication complexity of this protocol is 
$\calO (\log (n/\epsilon))$.
\end{lemma}

We note that this simulation can easily be extended to the case when the two input strings are not of the same length, by first communicating
the two lengths, and continuing only if they are equal. This leaves the communication complexity of the simulation protocol $O(\log (n/\epsilon))$
where $n = \text{max}(|x|,|y|)$.

We will use the following lemma. This lemma, and its proof, are implicit in~\cite{BBKLS}. We give here the proof for completeness.
\begin{lemma}[\cite{BBKLS}]\label{cor lcp}
For every input distribution $\mu$, and every positive error probability $\delta$,  any  protocol $\tilde{\rho}$ 
that uses the lcp box $\ell$ times on average (on the input distribution $\mu$ and the internal randomness of $\tilde{\rho}$)
on strings of length at most $C$, 
can be simulated  with error $\delta$ by a protocol $\rho$ that does not use an lcp box and communicates on average 
$O(\ell \log(\frac{\ell C}{\delta}))$  bits more than $\tilde{\rho}$.
\end{lemma}
\begin{proof}
The protocol  $\rho$ simulates $\tilde{\rho}$ by replacing each use of the lcp box with the protocol given by Lemma~\ref{le:lcp_simulation},
with error $\epsilon$, $\epsilon$ to be defined later.

Since each call to that protocol fails with probability at most $\epsilon$, the  (distributional) error introduced by the use of the simulation 
protocol instead of the lcp box is at most $\epsilon \ell$. We thus take $\epsilon=\delta / \ell$ and get that  the simulation fails with 
(distributional) probability $\delta$.

By Lemma~\ref{le:lcp_simulation} each call to the protocol simulatimg the lcp box has communication complexity $O(\log (C/\epsilon))$.
We get that on average $\rho$  sends $O(\ell \log (C/\epsilon))= O(\ell \log (\frac{\ell C}{\delta}))$ bits more than $\tilde{\rho}$.
\end{proof}

We  use the lcp box in the definition of the protocols in our proof, and then use Lemma~\ref{cor lcp} to obtain our final result at the end.
\begin{proof}[Proof of theorem \ref{compression}]
Fix the public randomness to be $r$.
For each $i$, define the set $\calX_i$ to be the set of possible inputs of player $i$, and the set $\Pi_{(i)}(x_i)$ to be the set of possible transcripts of player $i$, given that player $i$ has input $x_i$ (and the public randomness is $r$):
$$\Pi_{(i)}(x_i) = \overleftrightarrow{_{\!\;}\pi_i}(\calX_1,\dots,\calX_{i-1},x_i,\calX_{i+1},\dots,\calX_k, r)~.$$
The messages being self-delimiting and the protocol $\pi$ being oblivious, $\Pi_{(i)}(x_i)$ is naturally defined as a set of binary strings.

Each player $i$ can now represent $\Pi_{(i)}(x_i)$ by a binary tree $T_i$ as follows.
We note that actually computing $T_i$ takes exponential time. However, we are concerned with the communication complexity of the protocol and not by its computational complexity. 
\begin{enumerate}
\item The root is the largest common prefix (lcp) of the transcripts in $\Pi_{(i)}(x_i)$, and the remaining nodes are defined inductively.
\item For node $\tau$, we have  \begin{itemize}
\item the first child of $\tau$ is the lcp of the transcripts in $\Pi_{(i)}(x_i)$ beginning with $\tau\circ 0$, i.e., $\tau$ concatenated with the bit $0$.
\item the second child of $\tau$ is the lcp of the transcripts in $\Pi_{(i)}(x_i)$ beginning with $\tau\circ 1$.
\end{itemize}
\item The leaves are labelled by the possible transcripts of player $i$, i.e., the elements of $\Pi_{(i)}(x_i)$. \end{enumerate}

We define the {\em weight } of a  leaf $f$ with label $t_i$  to be
$$w(t_i) = \Pr\limits_{(X_j)_{j \neq i} \mid X_i = x_i}[\overleftrightarrow{_{\!\;}\pi_i}(X_1,\dots,X_{i-1},x_i,X_{i+1},\dots,X_k, r) = t_i]~.$$
The weight of a non-leaf node is defined by induction as the sum of the weights of its children. By construction, the weight of the root is $1$.

\begin{sloppypar}
We say that $(t_1,\dots,t_k) \in \Pi_{(1)}(x_1) \times \ldots \times \Pi_{(k)}(x_k)$ is a \textit{coherent profile} if every message from $i$ to $j$ appears with the same content in $t_i$ and $t_j$. In fact, given $(x_1,\ldots,x_k)$, 
the profile 
${(\overleftrightarrow{_{\!\;}\pi_1}(x_1,\dots,x_k, r),\dots,\overleftrightarrow{_{\!\;}\pi_k}(x_1,\dots,x_k, r))}$ is the only coherent profile. 
Assume towards a contradiction that there are two distinct coherent profiles, given $(x_1,\ldots,x_k)$.
 Each coherent profile gives rise to a transcript of the protocol.
Let $m$ be the first message, according to the global order of messages of an
oblivious protocol as defined in Section~\ref{sec:xor}, which is different in these two transcripts. But, 
each message sent from player $i$ to player $j$ is fully determined by the input $x_i$ 
and the previous messages according to that order (and the shared randomness), and thus $m$ cannot differ in the two transcripts.
\end{sloppypar}

We now define the protocol $\tilde\rho$ which allows 
the players  to collaborate and efficiently find  this coherent profile, i.e.,  protocol $\tilde\rho$ allows
 each player $i$ to find $\overleftrightarrow{_{\!\;}\pi_i}(x_1,\dots,x_k, r)$. 
 
 \noindent The players  proceed in stages $s=1,2\dots$. We will have the invariant that 
at the beginning of any stage $s$, each player $i$ is at a node $\tau_i(s)$ of its transcript tree $T_i$, such that $(\tau_1(s),\dots,\tau_k(s))$ is a (term-wise) prefix of 
$(\overleftrightarrow{_{\!\;}\pi_1}(x_1,\dots,x_k, r),\dots,\overleftrightarrow{_{\!\;}\pi_k}(x_1,\dots,x_k, r))$. 
At any time, given $\tau_i(s)$, for any $i$ and $s$, player $i$  furthermore has a candidate leaf 
$t_i(s)$ in the tree $T_i$ (representing a candidate for its transcript), defined as follows: 
player $i$ defines $\tau^1 = \tau_i(s)$, and then defines inductively $\tau^{j+1}$ to be the child of $\tau^j$ which has higher weight (breaking ties arbitrarily), until it reaches a leaf: this is the candidate $t_i(s)$. Observe that 
$t_i(s)$ is a descendent of $\tau_i(s)$ in $T_i$~\footnote{We define here a node to be a descendent of itself.} and that $t_i(s)$  corresponds to the transcript with highest probability conditioned on that the prefix of the transcript is the string corresponding to $\tau_i(s)$.  

At the beginning,
each player $i$ starts the protocol being at the node $\tau_i(1)$, which is the root of the tree $T_i$, and the invariant above clearly holds.
For each stage $s$ the players proceed as follows:
\begin{enumerate}
\item Each pair of players $(i,j)$  uses an lcp box to find the first occurrence where  the transcript between $i$ and $j$ 
in $t_i(s)$ is not coherent with the transcript between $i$ and $j$  in
$t_j(s)$. Let $q_{i,j}$ be the index of the message that includes this first occurrence, where the messages 
are numbered
according to the global order of all messages of an oblivious protocol as defined in Section~\ref{sec:xor}, and 
$\infty$ if
no such occurrence was found.
Let $Q_i=\min_j\{q_{i,j}\}$.
Observe that if for all pairs of players there is no such occurrence (i.e., $Q_i=\infty$ for all $i$), it means that $(t_1(s),\dots,t_k(s))$ is a coherent profile, each player $i$ has found $\overleftrightarrow{_{\!\;}\pi_i}(x_1,\dots,x_k, r)$.
\item\label{item:bcindex} Each player $i$ now broadcasts $Q_i$. Each player can then find $Q=\min_i\{Q_i\}$.
If $Q=\infty$, i.e., no pairwise inconsistency has been found between any two nodes, the protocol terminates 
and $(t_1(s),\ldots,t_k(s))$ is found as the coherent profile.  
\item\label{item:choice} 
Let $(i,j)$ be the pair of players such that $Q=q_{i,j}$.  The player who has the sender role of message number $Q$ is considered ``correct''.  Let this player be player $j$ and the player receiving the message, player $i$.
Player $i$ sets its $\tau_i(s+1)$: in $T_i$, starting from $t_i(s)$, it goes up the 
tree toward $\tau_i(s)$, until it reaches a node $\hat\tau_i$ which is correct (according to the result of the lcp box). Then, it defines $\tau_i(s+1)$ as the child of $\hat\tau_i$ which is not on the path from $\hat\tau_i$ to $t_i(s)$.
\item Any other player $j \neq i$ defines $\tau_j(s+1) = \tau_j(s)$. 
\end{enumerate}

We now claim by induction on the stages that the invariant stated above is preserved for all players at all times.
It clearly holds at the beginning. We claim that if it holds after stage $s$ then it also holds after stage $s+1$.
 For the $k-1$ players which define $\tau_j(s+1) = \tau_j(s)$ it clearly
continues to hold. For the single player, say player $i$,  which defines a new node as  $\tau_i(s+1)$ in Step~(\ref{item:choice}) we proceed as follows. 

We first claim, by induction on the index of the messages in the 
global order,  that for all messages with index $\ell < Q$, where message $\ell$ is sent from player $j$ to player $i$,
 it holds that the value of message number
 $\ell$ is the same in the coherent profile 
 $(\overleftrightarrow{_{\!\;}\pi_1}(x_1,\dots,x_k, r),\dots,\overleftrightarrow{_{\!\;}\pi_k}(x_1,\dots,x_k, r))$
  and in both $t_i(s+1)$ and $t_j(s+1)$.
The basis of the induction ($\ell=0$) clearly holds. The inductive step follows from observing that message
$\ell$  is fully determined by the input to player $j$ and the messages that appear before message 
$\ell$ in $\overleftrightarrow{_{\!\;}\pi_j}$. Thus, by the induction hypothesis the value of message $\ell$ in
 $t_j(s+1)$ is as it appears in the coherent profile  $(\overleftrightarrow{_{\!\;}\pi_1}(x_1,\dots,x_k, r),\dots,\overleftrightarrow{_{\!\;}\pi_k}(x_1,\dots,x_k, r))$. It follows from the definition 
 of $Q$ that the value of message $\ell$ is
  the same in $t_i(s+1)$ and $t_j(s+1)$. 
  
  For message $Q$, we have by similar arguments that its value according to $t_j(s+1)$ is consistent 
  with $\overleftrightarrow{_{\!\;}\pi_i}(x_1,\dots,x_k, r)$. The prefix of message $Q$ as appears in
  the  path from the root of  $T_i$ and delimited by $\tau_i(s+1)$ is consistent with $t_j(s+1)$ by the
   choice of $\tau_i(s+1)$ in Step~(\ref{item:choice}).

Now, since the relative order of messages in a transcript  $\overleftrightarrow{_{~}\Pi_i}$ and in the global order is the same, it follows that $\tau_i(s+1)$ represents a prefix of $\overleftrightarrow{_{\!\;}\pi_i}(x_1,\dots,x_k, r)$, as required.

 We now 
show that  for player $i$ which is the (single) player that  sets its $\tau_{i}(s+1)$ in Step (\ref{item:choice})
(i.e., the single player that changes its $\tau$ node and its guess of the transcript), 
$w(\tau_i(s+1)) \le \frac{1}{2} w(\tau_i(s))$. We look at the sequence $(\tau^j)$ defined by player $i$ when defining
its candidate leaf $t_i(s)$  as a function of $\tau_i(s)$. Let ${\tau^j}$ be the first common ancestor of $t_i(s)$ 
and $\tau_i(s+1)$. By construction, $\tau_i(s+1)$ is a child of $\tau^j$, and 
$t_i(s)$ is a descendant of the other child of ${\tau^j}$. 
 By the candidate leaf's construction process,
$w(\tau_i(s+1)) \le \frac{1}{2} w(\tau^j) \le \frac{1}{2} w(\tau_i(s))$.

We conclude the analysis. On inputs $(x_1,\dots,x_k)$, let $(t_1,\dots,t_k)$ denote the coherent profile. First note that 
with each stage the depth of one of the nodes $\tau_i$  increases. We proved that at any time $(\tau_1,\ldots,\tau_k)$
is a term-wise prefix of $(t_1,\dots,t_k)$. Thus (unless $\CC(\pi)$ is not finite, in which case the theorem trivially holds), the protocol terminates in finite time, with the ``candidate'' profile $(t_1,\dots,t_k)$. To give an upper bound on the  number of stages until this happens, observe that each player will set  its $\tau_i$  in 
Step (\ref{item:choice}) (i.e., will change its $\tau_i$) at most  $\log\frac{1}{w(t_i)}$ times, because the weight of the node $\tau_i$ at least halves with each such 
change (recall that the root has weight $1$). Since in each stage there is exactly one player that changes its  $\tau_i$, the total number of stages, $S$, is bounded from above by 
$\sum_{i=1}^k \log\frac{1}{w(t_i)}$.
We now take the average over inputs and over the shared randomness:\\

\scalebox{0.95}{\parbox{.5\linewidth}{
\begin{align*}
\Exp_{r,x} [S] &\le \Exp_{r,x} \sum\limits_{i=1}^k \log\frac{1}{w(t_i)}\\
&= \sum\limits_{i=1}^k \Exp_{r,x_i}\left[\Exp_{(x_j)_{j \neq i} \mid X_i = x_i}\left[\log\frac{1}{w(t_i)}\right]\right]\\
&= \sum\limits_{i=1}^k \Exp_{r,x_i}\left[\Exp_{(x_j)_{j \neq i} \mid X_i = x_i}\left[\log\frac{1}{\Pr\limits_{(X_j)_{j \neq i} \mid X_i = x_i}\left[\overleftrightarrow{_{\!\;}\pi_i}(X_1,\dots,X_{i-1},x_i,X_{i+1},\dots,X_k, r) = \overleftrightarrow{_{\!\;}\pi_i}(x_1,\dots,x_k,r)\right]}\right]\right]\\
&= \sum\limits_{i=1}^k \Exp_{r,x_i}\left[\Exp_{t_i \mid X_i = x_i, R=r}\left[\log\frac{1}{\Pr\limits_{(X_j)_{j \neq i} \mid X_i = x_i}\left[\overleftrightarrow{_{\!\;}\pi_i}(X_1,\dots,X_{i-1},x_i,X_{i+1},\dots,X_k, r) = t_i\right]}\right]\right]\\
&=\sum\limits_{i=1}^k \Exp_{r,x_i}\left[H(\overleftrightarrow{_{~}\Pi_i} \mid x_i r)\right] \\
&=\sum\limits_{i=1}^k H(\overleftrightarrow{_{~}\Pi_i} | X_i R^p) \\
&=\sum\limits_{i=1}^k I(X_{-i} ; \overleftrightarrow{_{~}\Pi_i} | X_i R^p) \\
&=\sum\limits_{i=1}^k I(X_{-i} ; \Pi_i | X_i R^p) \\
&=\IC_\mu(\pi)~,
\end{align*}}}

where the one before last equality follows from Proposition~\ref{eq def} (and the one-to-one correspondence between the definition of
$\overleftrightarrow{_{~}\Pi_i}$ used here and the definition of Section~\ref{sec:model}).
	
We have shown that the average number of stages is bounded by 
$\IC_\mu(\pi)$. At each stage, the communication consists of  $\frac{k(k-1)}{2}$ calls to the lcp box 
on strings of length at most $\calO(\CC(\pi))$ 
(one call for each pair of players), plus  $k(k-1)$ messages  of broadcasts of indices at Step~(\ref{item:bcindex}), each message of size $\calO(\log \CC(\pi))$. Hence we have a protocol with, on average,  $O(k^2 \cdot \IC_\mu(\pi))$ calls to the lcp box  on strings of length at 
most $\calO(\CC(\pi))$ and with 
$$\ACC_\mu(\tilde\rho) = \calO\left(k^2 \cdot \IC_\mu(\pi)\cdot\log(\CC(\pi))\right)~.$$

Using Lemma~\ref{cor lcp} we can replace each use of the lcp box with a simulation protocol, to get the protocol
 $\rho$ which simulates $\pi$ with distributional  error $\epsilon + \delta$ and average communication:
\begin{align*}
\ACC_\mu(\rho) &= \ACC_\mu(\tilde\rho) + \calO\left(k^2 \cdot \IC_\mu(\pi)\cdot\log\frac{k^2\cdot \IC_\mu(\pi)\cdot \CC(\pi)}{\delta}\right)\\
&= \calO\left(k^2 \cdot \IC_\mu(\pi)\cdot\log\frac{k^2 \cdot \IC_\mu(\pi) \cdot \CC(\pi)}{\delta}\right).
\end{align*}
\end{proof}

\subsection{A direct sum for $\PIC$ implies a direct sum for $\CC$}

The next theorem states that if $\PIC$ has a certain direct sum property then one can compress
the communication of certain multi-party protocols. Note that the result of this theorem is meaningful when $t$
 is large with respect to $k$.

\begin{theorem}\label{DS PIC}
In the oblivious setting, given a $k$-variable function $f$, if for any $t$ and any distribution $\mu$ on inputs of $f$
the existence of a protocol $\pi$ computing $f^{\otimes t}$ with error $\epsilon \ge 0$ implies that there exists a protocol $\pi'$ computing $f$ with error $\epsilon$ and satisfying $\PIC_\mu(\pi') \le \frac{1}{t} \PIC_{\mu^{\otimes t}}(\pi)$, $\CC(\pi') \le \CC(\pi)$, then for any fixed $\delta >0$,  for any $t$ 
$$
\CC^{2(\epsilon + \delta)}(f) =  \calO\left(\frac{1}{t(\epsilon + \delta)} \cdot k^3\cdot  \log(k)\cdot \CC^{\eps}(f^{\otimes t})\cdot \log\frac{k^2\cdot (\CC^{\eps}(f^{\otimes t})^2}{\delta}\right)~.
$$

\end{theorem}

To prove this theorem, we first need the following lemma.
\begin{lemma}\label{2 avg}
Given an input distribution $\mu$, any $k$-party protocol with distributional error $\frac{\epsilon}{2}$ and average communication complexity $C$ can be turned into an oblivious protocol with distributional error $\epsilon$ and worst case communication complexity $\frac{C \cdot k\cdot \log(k)}{\epsilon}$.
\end{lemma}
\begin{proof}
Let $\pi$ be a protocol with error $\frac{\epsilon}{2}$ and average communication complexity $C$. We  now define a protocol $\pi'$, which is similar to $\pi$  but where player $1$ acts as a ``coordinator'', in addition to his original role in $\pi$, and the other players can only communicate with player $1$. 

In $\pi'$ the players will receive the messages from their peers via the coordinator in a way to be described below. When they wish to send a message to a peer, they will add this message, as a string of bits, to a local queue, together with the destination of the message. They will send the messages 
to their peers via the coordinator in a way to be described below.

So that the players can send and receive the messages the coordinator (player $1$) imposes {\em phases} on the players.
In  every phase, player $1$ sends a message to all players indicating the beginning of the phase. Each player then takes the next {\em bit}, denote
it $b$, from  its local queue and sends to the coordinator the message  $(b,i)$, where $i$ is the destination of the the message $b$ is part of.
If the player has no bits in its queue it sends the message ``no'' to player $1$.
Player $1$, after having  received all $k-1$ messages, forwards the bits it received to the various players. Every player $i$, where at
least one bit destined to $i$ has been received, receives a message of the form $(b_1,j_1),\ldots(b_q,j_q)$ (encoded in a self-delimiting manner)
where $j_\ell$, $1 \leq \ell  \leq q$ denote the
origins of the bit, and all other players receive the message ``no''. Observe that the players receiving the bits in this way can reconstruct the 
messages of the protocol $\pi$ since all messages (of $\pi$) are self delimiting, and can thus locally run the original protocol $\pi$.
The protocol $\pi'$ consists of exactly $T = \left\lceil\frac{2C}{\epsilon}\right\rceil$ such phases. If at the end of $\pi'$ a certain player did
not output according to $\pi$, then in $\pi'$ that player outputs an arbitrary output.

Note that $\pi'$ is oblivious. Moreover, $\pi'$ fails to simulate $\pi$ (i.e., there is at least one player which outputs differently in $\pi$ and in $\pi'$)
only if $\pi'$ interrupts the simulation of $\pi$ at the end of the $T$'th phase. 
Since  every phase in $\pi'$ transmits at least one additional bit of the communication of $\pi$, the probability that $\pi'$ interrupts the simulation 
of $\pi$ is the probability that the communication cost of $\pi$ is more than $T$.
 We have 
$\Pr\limits_{x,r}(|\Pi(x)| \ge T) \le \frac{C}{T} \le \frac{\epsilon}{2}$ by Markov inequality. Adding that to the
 original error probability of $\pi$, we have that  the protocol $\pi'$ has error $\epsilon$.

Last, every phase in protocol $\pi'$ consists of communication $\calO(k\log(k))$, and protocol $\pi'$ thus has worst case communication $\calO(T \cdot k\cdot \log(k)) = \calO\left(\frac{C \cdot k\cdot\log(k)}{\epsilon}\right)$.
\end{proof}

\begin{proof}[Proof of Theorem \ref{DS PIC}]
Consider a protocol $\pi$ computing $f^{\otimes t}$ with error $\epsilon$. Let $\mu$ be a distribution on inputs of $f$. By hypothesis, there exist a protocol $\pi'$ computing $f$ with error $\epsilon$ and satisfying $\PIC_\mu(\pi') \le \frac{1}{t}\cdot \PIC_{\mu^{\otimes t}}(\pi)$, $\CC(\pi') \le \CC(\pi)$. By Theorem \ref{th:protocol_min_on_pub}, there exists such 
$\pi'$ that uses only public randomness.

Applying successively Theorem~\ref{compression pic}
 and Lemma~\ref{2 avg}, we get a protocol $\rho_\mu$ with distributional error $2(\epsilon + \delta)$ such that
\begin{align*}
\CC(\rho_\mu) 
&=\calO\left(\frac{1}{(\epsilon + \delta)}\cdot k^3\cdot \log(k)\cdot  \PIC_\mu(\pi')\cdot\log\frac{k^2 \cdot \PIC_\mu(\pi') \cdot \CC(\pi')}{\delta}\right) \\
&=\calO\left(\frac{1}{(\epsilon + \delta)}\cdot k^3\cdot \log(k)\cdot  \PIC_\mu(\pi')\cdot\log\frac{k^2 \cdot  (\CC(\pi'))^2}{\delta}\right)~.
\end{align*}
Thus,
\begin{align*}
\CC(\rho_\mu) &= \calO\left(\frac{1}{t(\epsilon + \delta)} \cdot k^3\cdot  \log(k)\cdot \PIC_{\mu^{\otimes t}}(\pi)\cdot \log\frac{k^2\cdot (\CC(\pi))^2}{\delta}\right)\\
&= \calO\left(\frac{1}{t(\epsilon + \delta)} \cdot k^3\cdot  \log(k)\cdot \CC(\pi)\cdot \log\frac{k^2\cdot (\CC(\pi))^2}{\delta}\right)~.
\end{align*}
Since the above holds  for any distribution $\mu$, the minimax theorem implies that 
$$
\CC^{2(\epsilon + \delta)}(f) =  \calO\left(\frac{1}{t(\epsilon + \delta)} \cdot k^3\cdot  \log(k)\cdot \CC^{\eps}(f^{\otimes t})\cdot \log\frac{k^2\cdot (\CC^{\eps}(f^{\otimes t}))^2}{\delta}\right)~.
$$

\end{proof}

\bibliography{bi}

\end{document}